\def\comments{1}

\documentclass[letterpaper,11pt]{article}

\usepackage{preamble}
\usepackage{bbm}
\usepackage[symbol]{footmisc}
\allowdisplaybreaks

\newcommand{\iprod}[1]{\langle #1 \rangle}

\newcommand{\llnorm}[1]{\left\lVert#1\right\rVert_2}

\newcommand{\abs}[1]{\left| #1 \right|}

\renewcommand{\epsilon}{\varepsilon}
\newcommand{\ball}[2]{\mathit{B}_{#2}\left( #1 \right)}

\newcommand{\ONE}{\mathbbm{1}}

\newcommand{\id}{\mathbb{I}}
\newcommand{\Lap}{\mathrm{Lap}}

\newcommand{\OPT}{\mathrm{OPT}}
\newcommand{\Score}{\textsc{Score}}
\newcommand{\Match}{\mathrm{Match}}
\newcommand{\Median}{\mathrm{Median}}

\newcommand{\Win}{\mathsf{Win}}
\newcommand{\Lose}{\mathsf{Lose}}
\newcommand{\Tie}{\mathsf{Tie}}

\newcommand{\PDPRE}{\ensuremath{\textsc{PDPRE}}}
\newcommand{\CDPRE}{\ensuremath{\textsc{CDPRE}}}
\newcommand{\ADPRE}{\ensuremath{\textsc{ADPRE}}}
\newcommand{\PDPODME}{\ensuremath{\textsc{PDPODME}}}
\newcommand{\CDPODME}{\ensuremath{\textsc{CDPODME}}}
\newcommand{\ADPODME}{\ensuremath{\textsc{ADPODME}}}
\newcommand{\PDPHDME}{\ensuremath{\textsc{PDPHDME}}}
\newcommand{\CDPHDME}{\ensuremath{\textsc{CDPHDME}}}

\title{Private Mean Estimation of Heavy-Tailed Distributions}
\author{
Gautam Kamath\thanks{Cheriton School of Computer Science, University of Waterloo. {\tt g@csail.mit.edu}. }
\and
Vikrant Singhal\thanks{Khoury College of Computer Sciences, Northeastern University. {\tt singhal.vi@northeastern.edu}. }
\and
Jonathan Ullman\thanks{Khoury College of Computer Sciences, Northeastern University. {\tt jullman@ccs.neu.edu}. }
}

\begin{document}

\maketitle \footnotetext{Authors are in alphabetical order.}

\begin{abstract}
We give new upper and lower bounds on the minimax sample complexity of differentially private mean estimation of distributions with bounded $k$-th moments.
  Roughly speaking, in the univariate case, we show that $$n = \Theta\left(\frac{1}{\alpha^2} + \frac{1}{\alpha^{\frac{k}{k-1}}\varepsilon}\right)$$ samples are necessary and sufficient to estimate the mean to $\alpha$-accuracy under $\varepsilon$-differential privacy, or any of its common relaxations.
  This result demonstrates a qualitatively different behavior compared to estimation absent privacy constraints, for which the sample complexity is identical for all $k \geq 2$.
  We also give algorithms for the multivariate setting whose sample complexity is a factor of $O(d)$ larger than the univariate case. 
\end{abstract}
\newpage
\tableofcontents
\newpage

\section{Introduction}
Given samples $X_1, \dots, X_n$ from a distribution $\mathcal{D}$, can we estimate the mean of $\mathcal{D}$?
This is the problem of \emph{mean estimation} which is, alongside hypothesis testing, one of the most fundamental questions in statistics.
As a result, answers to this problem are known in fairly general settings.
For instance, the empirical mean is known to be an optimal estimate of a distribution's true mean under minimal assumptions.

That said, statistics like the empirical mean put aside any concerns related to the sensitivity, and might vary significantly based on the addition of a single datapoint in the dataset.
While this is not an inherently negative feature, it becomes a problem when the dataset contains personal information, and large shifts based on a single datapoint could potentially violate the corresponding individual's \emph{privacy}.
In order to assuage these concerns, we consider the problem of mean estimation under the constraint of \emph{differential privacy} (DP)~\cite{DworkMNS06}, considered by many to be the gold standard of data privacy.
Informally, an algorithm is said to be differentially private if its distribution over outputs is insensitive to the addition or removal of a single datapoint from the dataset.  Differential privacy has enjoyed widespread adoption, including deployment in by Apple~\cite{AppleDP17}, Google~\cite{ErlingssonPK14}, Microsoft~\cite{DingKY17}, and the US Census Bureau for the 2020 Census~\cite{DajaniLSKRMGDGKKLSSVA17}.

In this vein, a recent line of work~\cite{KarwaV18, KamathLSU19, BunKSW19} gives nearly optimal differentially private algorithms for mean estimation of sub-Gaussian random variables.  Roughly speaking, to achieve accuracy $\alpha$ under $\varepsilon$-differential privacy in a $d$-dimensional setting, one requires $n = \tilde O(\frac{d}{\alpha^2} + \frac{d}{\alpha\varepsilon})$ samples, a mild cost of privacy over the non-private sample complexity of $O(\frac{d}{\alpha^2})$, except when $\eps$ is very small (corresponding to a very high level of privacy).  However, these results all depend on the strong assumption that the underlying distribution being sub-Gaussian.  Indeed, many sources of data in the real world are known to be heavy-tailed in nature, and thus we require algorithms which are effective even under these looser restrictions.  Thus, the core question of this work is
\begin{quote}
\emph{What is the cost of privacy when estimating the mean of heavy-tailed distributions?}
\end{quote} 

We make progress on this question by giving both algorithms and lower bounds for differentially private mean estimation on distributions with bounded $k$-th moments, for $k \geq 2$.  In particular, for univariate distributions, we show that the optimal worst-case sample complexity depends critically on the choice of $k$, which is qualitatively different from the non-private case.

\subsection{Results, Techniques, and Discussion}
In this section, we will assume familiarity with some of the most common notions of differential privacy: pure $\varepsilon$-differential privacy, $\rho$-zero-concentrated differential privacy, and approximate $(\varepsilon, \delta)$-differential privacy.
In particular, one should know that these are in (strictly) decreasing order of strength, formal definitions appear in Section~\ref{sec:preliminaries}.

We first focus on the univariate setting, proving tight upper and lower bounds for estimation subject to bounds on every possible moment.
\begin{theorem} [Theorems~\ref{thm:one-d-pdp} and~\ref{thm:one-d-pdp-lb}] \label{thm:main-univariate}
For every $k \geq 2$, $0 < \eps,\alpha < 1$, and $R > 1$, there is an $\eps$-DP algorithm that takes
$$
n = O\left(
\frac{1}{\alpha^2} +
\frac{1}{\eps \alpha^{\frac{k}{k-1}}} +
\frac{\log(R)}{\eps}\right)
$$
samples from an arbitrary distribution $\cD$ with mean $\mu$ such that $\mu \in (-R,R)$ and $\ex{}{| \cD- \mu|^k} \leq 1$ and returns $\hat\mu$ such that, with high probability, $| \hat\mu - \mu | \leq \alpha$.  Moreover, any such $\eps$-DP algorithm requires $n = \Omega(\frac{1}{\alpha^2} +
\frac{1}{\eps \alpha^{\frac{k}{k-1}}} + \frac{\log(R)}{\eps})$ samples in the worst case.\footnote{Analogous tight bounds hold for zCDP and $(\eps,\delta)$-DP, and these bounds differ only in the dependence on $R$ in the final term.  In particular, $\Omega(1/\alpha^2 + 1/\eps \alpha^{k/(k-1)})$ samples are necessary for any of the variants of differential privacy.}
\end{theorem}

\begin{remark}
    In Theorem~\ref{thm:one-d-pdp} we chose to
    reduce the number of parameters by making
    only the assumption that $\ex{}{|\cD - \mu|^k} \leq 1$,
    and bounding the absolute error $ | \hat\mu - \mu |$.
    More generally, we can consider a setting
    where the variance is $\sigma^2 = \ex{}{(\cD - \mu)^2}$
    and the $k$-th moment satisfies
    $\ex{}{| \cD- \mu |^{k}} \leq M^{k}\sigma^{k}$
    for some $M \geq 1$, and we want to bound the
    normalized error $| \hat\mu - \mu | / \sigma$.
    It is without loss of generality to solve the
    simplified problem.  For example, if the standard
    deviation is known, then we can renormalize
    the data by a factor of $M\sigma$, after which
    the distribution has standard deviation $1/M$
    and $k$-th moment at most $1$.  Now we can apply
    our theorem with accuracy $\alpha' = \alpha / M$,
    in which case we get a sample complexity of
    $O(M^{k/(k-1)}/\eps \alpha^{k/(k-1)})$.
\end{remark}

Note that, absent privacy constraints, the sample complexity of mean estimation with bounded $k$-th moments is $n = O(1/\alpha^2)$ samples, for any $k \geq 2$.
However, if we require the algorithm to be differentially private, there is a qualitatively different picture in which the cost of privacy decays as we have stronger bounds on the moments of the distribution.
Our upper bounds follow a noised and truncated-empirical-mean approach.
While this is similar to prior work on private mean estimation~\cite{KarwaV18, KamathLSU19, CaiWZ19, BunS19}, we must be more aggressive with our truncation than before.
In particular, for the Gaussian case, strong tail bounds allow one to truncate in a rather loose window and not remove any points if the data was actually sampled from a Gaussian.
Since we consider distributions with much heavier tails, trying to not discard any points would result in a very wide truncation window, necessitating excessive amounts of noise.
Instead, we truncate in a way that balances the two sources of error: bias due to valid points being discarded, and the magnitude of the noise due to the width of the truncation window.
To be a bit more precise, our setting of parameters for truncation can be viewed in two different ways: either we truncate so that (in expectation) $1/\varepsilon$ points are removed, and we require $n$ to be large enough to guarantee accuracy, or we truncate so that $\alpha^{k/ (k-1)}$ probability mass is removed, and we require $n$ to be large enough to guarantee privacy.
These two perspectives on truncation are equivalent when $n$ is at the critical value that makes up our sample complexity.

Our lower bound is proved via hypothesis testing.  We demonstrate that two distributions that satisfy the conditions and are indistinguishable with fewer than the prescribed number of samples.
Due to an equivalence between pure and approximate differential privacy in this setting, our lower bounds hold for the most permissive privacy notion of $(\varepsilon,\delta)$-DP, even for rather large values of $\delta$.

Turning to the multivariate setting, we provide separate algorithms for concentrated and pure differential privacy, both of which come at a multiplicative cost of $O(d)$ in comparison to the univariate setting.
We state the concentrated DP result first.
\begin{theorem} [Theorem~\ref{thm:high-d-cdp}]
  \label{thm:high-d-cdp-intro}
  For every $d$, $k \geq 2$, $\eps,\alpha > 0$, and $R > 1$, there is a polynomial-time $\frac{\eps^2}{2}$-zCDP algorithm that takes
   $$n \geq O\left(
  \frac{d}{\alpha^2} +
  \frac{d}{\eps \alpha^{\frac{k}{k-1}}} +
    \frac{ \sqrt{d\log(R)}\log(d)}{\eps}
  \right)
  $$
  samples from an arbitrary distribution $\cD$ on $\R^d$ with mean vector $\mu$ such that $\| \mu \|_2 \leq R$ and bounded $k$-th moments $\sup_{v \in \mathbb{S}^{d-1}} \ex{}{|\langle v, \cD- \mu \rangle|^k} \leq 1$ and returns $\hat\mu$ such that, with high probability, $\| \hat\mu - \mu \|_2 \leq \alpha$.  
\end{theorem}

Similar to the univariate case, we rely upon a noised and truncated empirical mean (with truncation to an $\ell_2$ ball).
The computations required to bound the bias of the truncated estimator are somewhat more involved and technical than the univariate case.

Our pure-DP multivariate mean estimator has the following guarantees.
\begin{theorem}[Theorem~\ref{thm:high-d-pdp}]
  \label{thm:high-d-pure-intro}
For every $d$, $k \geq 2$, $\eps,\alpha > 0$, and $R > 1$, there is a (possibly exponential time) pure $\eps$-DP algorithm that takes
$$n \geq O\left(
\frac{d}{\alpha^2} +
\frac{d}{\eps \alpha^{\frac{k}{k-1}}} +
\frac{ d\log(R) \log(d)}{\eps}
\right)
$$
samples from an arbitrary distribution $\cD$ on $\R^d$ with mean vector $\mu$ such that $\| \mu \|_2 \leq R$ and bounded $k$-th moments $\sup_{v \in \mathbb{S}^{d-1}} \ex{}{|\langle v, \cD- \mu \rangle|^k} \leq 1$ and returns $\hat\mu$ such that, with high probability, $\| \hat\mu - \mu \|_2 \leq \alpha$.  
\end{theorem}
We discuss the similarities and differences between Theorems~\ref{thm:high-d-cdp-intro} and~\ref{thm:high-d-pure-intro}.
First, we note that the first two terms in the sample complexity are identical, similar to the multivariate Gaussian case, where distribution estimation under pure and concentrated DP share the same sample complexity~\cite{KamathLSU19, BunKSW19}.
This is contrary to certain problems in private mean estimation, where an $O(\sqrt{d})$ factor often separates the two complexities~\cite{BunUV14, SteinkeU15, DworkSSUV15}.
It appears that these qualitative gaps may or may not arise depending on the choice of norm and the assumptions we put on the underlying distribution (see Section 1.1.4 of~\cite{KamathLSU19} and Remark 6.4 of~\cite{BunKSW19}) for more discussion.
We point out that the estimator of Theorem~\ref{thm:high-d-pure-intro} is not computationally efficient, while the estimator of Theorem~\ref{thm:high-d-cdp-intro} is.
However, even for the well structured Gaussian case, no computationally-efficient algorithm is known under pure DP~\cite{KamathLSU19, BunKSW19}.

Technically, our multivariate-pure-DP algorithm is quite different from our other algorithms.
It bears significant resemblance to approaches based on applying the ``Scheff\'e estimator'' to a cover for the family of distributions~\cite{Yatracos85, DevroyeL96, DevroyeL97, DevroyeL01}, and also a tournament-based approach of Lugosi and Mendelson~\cite{LugosiM19a} for non-private mean estimation with sub-Gaussian rates.
These approaches reduce an estimation problem to a series of pairwise comparisons (i.e., hypothesis tests). 
We cover the space of candidate means, and perform a series of tests of the form ``Which of these two candidates is a better fit for the distribution's mean?''
As mentioned before, there are often gaps between our understanding of multivariate estimation under pure and concentrated DP, and the primary reason is that the Laplace and Gaussian mechanisms have sensitivities based on the $\ell_1$ and $\ell_2$ norms, respectively.
We avoid paying the extra $O(\sqrt{d})$ which often arises in the multivariate setting by reducing to a series of \emph{univariate} problems---given two candidate means, we can project the problem onto the line which connects the two.
By choosing whichever candidate wins all of its comparisons, we can get an accurate estimate for the mean overall.
Crucially, using techniques from~\cite{BunKSW19}, we only pay logarithmically in the size of the cover.

\begin{remark}
    In Theorems~\ref{thm:high-d-cdp} and~\ref{thm:high-d-pdp} we
    use the standard formulation of bounded moments
    for distributions on $\R^d$, which means that
    for every direction $v$, the univariate distribution
    obtained by projecting onto $v$ has bounded $k$-th
    moment. Although this is the standard definition
    of bounded moments for multivariate distributions,
    one could potentially consider other classes of
    heavy-tailed distributions, for example, one which bounds $\ex{}{\|\cD-\mu\|_2^k}$.
\end{remark}

Finally, we prove some lower bounds for multivariate private mean estimation.
\begin{theorem} [Theorem~\ref{thm:high-d-lb}]
Any pure $\eps$-DP algorithm that takes samples from an arbitrary distribution on $\R^d$ with bounded 2nd moments and returns $\hat\mu$ such that $\| \hat\mu - \mu\|_2 \leq \alpha$ requires
$
n = \Omega\left( \frac{d}{\eps \alpha^2} \right)
$
samples from $\cD$ in the worst case.
\end{theorem}

After the original appearance of this work, it was
brought to our attention that more general lower
bounds exist in a work by Barber and Duchi~\cite{BarberD14}.
While their setting is slightly different than ours
(see Section~\ref{sec:related} for additional discussion),
their results imply the following (rephrased) lower bounds
in our setting.
\begin{theorem}[Proposition~4 of~\cite{BarberD14}]
  \label{thm:bd14}
    Suppose $\cA$ is an $(\eps,0)$-DP algorithm and
    $n \in \N$ is a number such that, for every
    distribution $\cD$ on $\R^d$, such that $\ex{}{\cD} = \mu$
    and $\sup_{v:\|v\|=1}\ex{}{|\langle v,\cD-\mu\rangle|^k} \leq 1$,
    $$\ex{X_1,\dots,X_n\sim\cD,\cA}{\|\cA(X)-\mu\|^2_2}\leq\alpha^2.$$
    Then $n=\Omega\left(\tfrac{d}{\eps\alpha^{\frac{k}{k-1}}}\right)$.
\end{theorem}

Note that the sample complexity matches the upper bounds in Theorem~\ref{thm:high-d-pure-intro}, showing that our algorithms are optimal for every $k \geq2$.

\subsection{Related Work}
\label{sec:related}
The most closely related works to ours are~\cite{BlumDMN05, BunUV14, SteinkeU15, DworkSSUV15, SteinkeU17a, KarwaV18, KamathLSU19, CaiWZ19, BunS19, AvellaMedinaB19, DuFMBG20, BiswasDKU20}, which study differentially private estimation of the mean of a distribution.
Some of these focus on restricted cases, such as product distributions or sub-Gaussians, which we generalize by making weaker moment-based assumptions.
Some instead study more general cases, including unrestricted distributions over the hypercube -- by making assumptions on the moments of the generating distributions, we are able to get better sample complexities.
The work of Bun and Steinke~\cite{BunS19} explicitly studies mean estimation of distributions with bounded second moment, but their sample complexity can be roughly stated as $O(1/\alpha^2\varepsilon^2)$, whereas we prove a tight bound of $\Theta(1/\alpha^2\varepsilon)$.
Furthermore, we go beyond second-moment assumptions, and show a hierarchy of sample complexities based on the number of moments which are bounded.
Some very recent works~\cite{DuFMBG20, BiswasDKU20} focus on designing practical tools for private estimation of mean and covariance, in univariate and multivariate settings.

After the original appearance of this work, a closely related prior work of Barber and Duchi was brought to our attention~\cite{BarberD14}.
They too study mean estimation of distributions under bounded moment assumptions, but with a qualitatively different definition of what it means for moments to be bounded.
In particular,
they assume that a distribution with mean $\mu$ has bounded $k$th moments
if $\ex{}{\llnorm{X-\mu}^k} \leq 1$, while we say that a distribution has bounded $k$th moments if 
$\ex{}{\abs{\iprod{X-\mu,v}}^k} \leq 1$
for any unit vector $v$. These two moment conditions are on
different scales. For example, in the case of a spherical
Gaussian in $d$ dimensions with covariance matrix $I$, our moment
condition says that the second moment is $1$, while
theirs says that it is $d$. It is easy to show that
their condition implies our condition in general. Therefore, their lower
bounds imply lower bounds under our definition, and we state the implications in Theorem~\ref{thm:bd14}. 
They also prove upper bounds for private mean estimation in their setting, which do not have similar strong implications for our setting.
Under pure DP, our upper bounds match their lower bounds, which shows
optimality of our results. We independently show the same
lower bound for the special case of $k=2$.
Another interesting feature of their results is the
separation between their approximate DP and pure DP upper
bounds. This, as we conjecture, does not exist under our
moment condition, since it is true in case of high-dimensional
Gaussians.

Among other problems, Duchi, Jordan, and Wainwright~\cite{DuchiJW13, DuchiJW17} study univariate mean estimation with moment bounds under the stricter constraint of \emph{local} differential privacy.
Our univariate results and techniques are similar to theirs: morally the same algorithm and lower-bound construction works in both the local and central model.
Translating their results to compare to ours, they show that the sample complexity of mean estimation in the local model is $O(1/\alpha^\frac{2k}{k-1}\varepsilon^2)$, the square of the ``second term'' in our sample complexity for the central model.
However, their investigation is limited to the univariate setting, while we provide new algorithms and lower bounds for the multivariate setting as well.
There has also been some work on locally private mean estimation in the Gaussian case~\cite{GaboardiRS19, JosephKMW19}.

This is just a small sample of work in differentially private distribution estimation, and there has been much study into learning distributions beyond mean estimation.
These are sometimes (but not always) equivalent problems -- for instance, learning the mean of a Gaussian distribution with known covariance is equivalent to learning the distribution in total variation distance.
Diakonikolas, Hardt, and Schmidt~\cite{DiakonikolasHS15} gave algorithms for learning structured univariate distributions.
Privately learning mixtures of Gaussians was considered in~\cite{NissimRS07,KamathSSU19}.
Bun, Nissim, Stemmer, and Vadhan~\cite{BunNSV15} give an algorithm for learning distributions in Kolmogorov distance.
Acharya, Kamath, Sun, and Zhang~\cite{AcharyaKSZ18} focus on estimating properties of a distribution, such as the entropy or support size.
Smith~\cite{Smith11} gives an algorithm which allows one to estimate asymptotically normal statistics with optimal convergence rates, but no finite sample complexity guarantees.
Bun, Kamath, Steinke, and Wu~\cite{BunKSW19} give general tools for private hypothesis selection and apply this to learning many distribution classes of interest.
For further coverage of differentially private statistics, see~\cite{KamathU20}.

In the non-private setting, there has recently been significant work in mean estimation of distributions with bounded second moments, in the ``high probability'' regime. 
That is, we wish to estimate the mean of a distribution with probability $1- \beta$, where $\beta >0$ might be very small.
While the empirical mean is effective in the sub-Gaussian case, more advanced techniques are necessary to achieve the right dependence on $1/\beta$ when we only have a bound on the second moment.
A recent series of papers has focused on identifying effective methods and making them computationally efficient~\cite{LugosiM19a, Hopkins18, CherapanamjeriFB19, DepersinL19, LugosiM19b}.
This high-probability consideration is not the focus of the present work, though we note that, at worst, our estimators incur a multiplicative factor of $\log (1/\beta)$ in achieving this guarantee.
We consider determining the correct dependence on the failure probability with privacy constraints an interesting direction for future study.
Interestingly, the tournament-based approach of Lugosi and Mendelson~\cite{LugosiM19a} is similar to our pure-DP mean estimation algorithm, though there are details which distinguish the two due to the very different settings.

Furthermore, this pure-DP mean estimation algorithm bears a significant resemblance to methods in a line on hypothesis selection, reducing to pairwise comparisons using the Scheff\'e estimator.
This style of approach was pioneered by Yatracos~\cite{Yatracos85}, and refined in subsequent work by Devroye and Lugosi~\cite{DevroyeL96, DevroyeL97, DevroyeL01}.
After this, additional considerations have been taken into account, such as computation, approximation factor, robustness, and more~\cite{MahalanabisS08, DaskalakisDS12b, DaskalakisK14, SureshOAJ14, AcharyaJOS14b, DiakonikolasKKLMS16, AcharyaFJOS18, BousquetKM19, BunKSW19}.

\section{Preliminaries}
\label{sec:preliminaries}

We formally state what it means for a distribution
to have its $k^{\text{th}}$ moment bounded.

\begin{defn}\label{defn:bounded-moment}
    Let $\cD$ be a distribution over $\R^d$ with mean
    $\mu$. We say that for $k \geq 2$, the $k^{\text{th}}$
    moment of $\cD$ is bounded by $M$, if for
    every unit vector $v \in \mathbb{S}^{d-1}$,
    $$\ex{}{\abs{\iprod{X-\mu,v}}^k} \leq M.$$
\end{defn}

Also, we define $\ball{p}{r} \subset \R^d$ to be
the ball of radius $r>0$ centered at $p \in \R^d$.

\subsection{Privacy Preliminaries}

\begin{defn}[Differential Privacy (DP) \cite{DworkMNS06}]
    \label{def:dp}
    A randomized algorithm $M:\cX^n \rightarrow \cY$
    satisfies $(\eps,\delta)$-differential privacy
    ($(\eps,\delta)$-DP) if for every pair of
    neighboring datasets $X,X' \in \cX^n$
    (i.e., datasets that differ in exactly one entry),
    $$\forall Y \subseteq \cY~~~
        \pr{}{M(X) \in Y} \leq e^{\eps}\cdot
        \pr{}{M(X') \in Y} + \delta.$$
    When $\delta = 0$, we say that $M$ satisfies
    $\eps$-differential privacy or pure differential
    privacy.
\end{defn}

\begin{defn}[Concentrated Differential Privacy (zCDP)~\cite{BunS16}]
    A randomized algorithm $M: \cX^n \rightarrow \cY$
    satisfies \emph{$\rho$-zCDP} if for
    every pair of neighboring datasets $X, X' \in \cX^n$,
    $$\forall \alpha \in (1,\infty)~~~D_\alpha\left(M(X)||M(X')\right) \leq \rho\alpha,$$
    where $D_\alpha\left(M(X)||M(X')\right)$ is the
    $\alpha$-R\'enyi divergence between $M(X)$ and
    $M(X')$.\footnote{Given two probability distributions
    $P,Q$ over $\Omega$,
    $D_{\alpha}(P\|Q) = \frac{1}{\alpha - 1}
    \log\left( \sum_{x} P(x)^{\alpha} Q(x)^{1-\alpha}\right)$.}
\end{defn}

Note that zCDP and DP are on different scales, but are otherwise can be ordered from most-to-least restrictive.  Specifically, $(\eps,0)$-DP implies $\frac{\rho^2}{2}$-zCDP, which implies $(\eps \sqrt{\log(1/\delta)}, \delta)$-DP for every $\delta > 0$~\cite{BunS16}.

\medskip Both these definitions are closed under post-processing and can be composed with graceful degradation of the privacy parameters.
\begin{lemma}[Post Processing \cite{DworkMNS06,BunS16}]\label{lem:post-processing}
    If $M:\cX^n \rightarrow \cY$ is
    $(\eps,\delta)$-DP, and $P:\cY \rightarrow \cZ$
    is any randomized function, then the algorithm
    $P \circ M$ is $(\eps,\delta)$-DP.
    Similarly if $M$ is $\rho$-zCDP then the algorithm
    $P \circ M$ is $\rho$-zCDP.
\end{lemma}

\begin{lemma}[Composition of DP~\cite{DworkMNS06, DworkRV10, BunS16}]\label{lem:composition}
    If $M$ is an adaptive composition of differentially
    private algorithms $M_1,\dots,M_T$, then the following
    all hold:
    \begin{enumerate}
        \item If $M_1,\dots,M_T$ are
            $(\eps_1,\delta_1),\dots,(\eps_T,\delta_T)$-DP
            then $M$ is $(\eps,\delta)$-DP for
            $$\eps = \sum_t \eps_t~~~~\textrm{and}~~~~\delta = \sum_t \delta_t.$$
        \item If $M_1,\dots,M_T$ are
            $(\eps_0,\delta_1),\dots,(\eps_0,\delta_T)$-DP
            for some $\eps_0 \leq 1$, then for every $\delta_0 > 0$, $M$
            is $(\eps, \delta)$-DP for
            $$\eps = \eps_0 \sqrt{6 T \log(1/\delta_0)}~~~~
                \textrm{and}~~~~\delta = \delta_0 + \sum_t \delta_t$$
        \item If $M_1,\dots,M_T$ are $\rho_1,\dots,\rho_T$-zCDP
            then $M$ is $\rho$-zCDP for $\rho = \sum_t \rho_t$.
    \end{enumerate}
\end{lemma}

\subsection{Basic Differentially Private Mechanisms.}
We first state standard results on achieving
privacy via noise addition proportional to
sensitivity~\cite{DworkMNS06}.

\begin{defn}[Sensitivity]
    Let $f : \cX^n \to \R^d$ be a function,
    its \emph{$\ell_1$-sensitivity} and
    \emph{$\ell_2$-sensitivity} are
    $$\Delta_{f,1} = \max_{X \sim X' \in \cX^n} \| f(X) - f(X') \|_1
    ~~~~\textrm{and}~~~~\Delta_{f,2} = \max_{X \sim X' \in \cX^n} \| f(X) - f(X') \|_2,$$
    respectively.
    Here, $X \sim X'$ denotes that $X$ and $X'$
    are neighboring datasets (i.e., those that
    differ in exactly one entry).
\end{defn}

For functions with bounded $\ell_1$-sensitivity,
we can achieve $\eps$-DP by adding noise from
a Laplace distribution proportional to
$\ell_1$-sensitivity. For functions taking values
in $\R^d$ for large $d$ it is more useful to add
noise from a Gaussian distribution proportional
to the $\ell_2$-sensitivity, to get $(\eps,\delta)$-DP
and $\rho$-zCDP.

\begin{lem}[Laplace Mechanism] \label{lem:laplacedp}
    Let $f : \cX^n \to \R^d$ be a function
    with $\ell_1$-sensitivity $\Delta_{f,1}$.
    Then the Laplace mechanism
    $$M(X) = f(X) + \Lap\left(\frac{\Delta_{f,1}}
        {\eps}\right)^{\otimes d}$$
    satisfies $\eps$-DP.
\end{lem}

\begin{lem}[Gaussian Mechanism] \label{lem:gaussiandp}
    Let $f : \cX^n \to \R^d$ be a function
    with $\ell_2$-sensitivity $\Delta_{f,2}$.
    Then the Gaussian mechanism
    $$M(X) = f(X) + \cN\left(0,\left(\frac{\Delta_{f,2}
        \sqrt{2\ln(2/\delta)}}{\eps}\right)^2 \cdot \id_{d \times d}\right)$$
    satisfies $(\eps,\delta)$-DP.
    Similarly, the Gaussian mechanism
    $$M_{f}(X) = f(X) +
        \cN\left(0, \left(\frac{\Delta_{f,2}}{\sqrt{2\rho}}\right)^2 \cdot \id_{d \times d}\right)$$
    satisfies $\rho$-zCDP.
\end{lem}

\begin{lem}[Private Histograms]\label{lem:priv-hist}
    Let $(X_1,\dots,X_n)$ be samples in some data universe
    $U$, and let $\Omega = \{h_u\}_{u \subset U}$
    be a collection of disjoint histogram buckets over $U$.
    Then we have $\eps$-DP, $\rho$-zCDP, and $(\eps,\delta)$-DP
    histogram algorithms with the following guarantees.
    \begin{enumerate}
        \item $\eps$-DP:
            $\ell_\infty$ error - $O\left(\tfrac{\log(\abs{U}/\beta)}{\eps}\right)$
            with probability at least $1-\beta$;
            run time - $\poly(n,\log(\abs{U}/\eps\beta))$
        \item $\rho$-zCDP:
            $\ell_\infty$ error - $O\left(\sqrt{\frac{\log(\abs{U}/\beta)}{\rho}}\right)$
            with probability at least $1-\beta$;
            run time - $\poly(n,\log(\abs{U}/\rho\beta))$
        \item $(\eps,\delta)$-DP:
            $\ell_\infty$ error - $O\left(\tfrac{\log(1/\delta\beta)}{\eps}\right)$
            with probability at least $1-\beta$;
            run time - $\poly(n,\log(\abs{U}/\eps\beta))$
    \end{enumerate}
\end{lem}

Part~1 follows from \cite{BalcerV19}. Part~2
follows trivially by using the Gaussian Mechanism
(Lemma~\ref{lem:gaussiandp}) instead of the Laplace
Mechanism (Lemma~\ref{lem:laplacedp}) in Part~1.
Part~3 holds due to \cite{BunNS16,Vadhan17}.
    
Finally, we recall the widely used \emph{exponential mechanism}.  
    
\begin{lem}[Exponential Mechanism \cite{McSherryT07}]
    \label{lem:exp-mechanism}
    The exponential mechanism $\cM_{\eps,S,\Score}(X)$
    takes a dataset $X \in \cX^n$, computes a score ($\Score : \cX^n \times S \to \R$)
    for each $p \in S$ with respect to $X$, and
    outputs $p \in S$ with probability proportional
    to $\exp\left(\frac{\eps \cdot \Score(X,p)}{2 \cdot \Delta_{\Score,1}}\right)$,
    where
    $$\Delta_{\Score,1} = \max\limits_{p \in S}\max\limits_{X\sim X' \in \cX^n}
        {\abs{\Score(X,p)-\Score(X',p)}}.$$
    It satisfies the following.
    \begin{enumerate}
        \item $\cM$ is $\eps$-differentially private.
        \item Let $\OPT_\Score(X) = \max\limits_{p \in S}\{\Score(X,p)\}$.
            Then
            $$\pr{}{\Score(X,\cM_{\eps,S,\Score}(X)) \leq
                \OPT_\Score(X) - \frac{2\Delta_{\Score,1}}{\eps}(\ln(\abs{S} + t))}
                \leq e^{-t}.$$
    \end{enumerate}
\end{lem}

\section{Estimating in One Dimension}  \label{sec:one-d}

In this section, we discuss estimating the mean
of a distribution whilst ensuring pure DP.
Obtaining CDP and approximate DP algorithms for
this is trivial, once we have the algorithm for
pure DP, as shall be discussed towards the end
of the upper bounds section. Finally, we show
that our upper bounds are optimal.

\subsection{Technical Lemmata}

Here, we lay out the two main technical lemmata
that we would use to prove our main results for
the section. The first lemma says that if we truncate
the distribution to within a large interval that is
centered close to the mean, then the mean of this
truncated distribution will be close to the original
mean.

\begin{lem}\label{lem:one-d-trunc-mean}
    Let $\cD$ be a distribution over $\R$ with
    mean $\mu$, and $k^{\text{th}}$ moment bounded
    by $1$. Let $\rho \in \R$, $0 < \tau < \tfrac{1}{16}$,
    and $\xi = \tfrac{C}{\tau^{\frac{1}{k-1}}}$ for a
    constant $C \geq 6$.
    Let $X \sim \cD$, and $Z$ be the following random
    variable.
    \[
    Z =
    \begin{cases}
        \rho - \xi & \text{if $X < \rho - \xi$}\\
        X & \text{if $\rho - \xi \leq X \leq \rho + \xi$}\\
        \rho + \xi & \text{if $X > \rho + \xi$}
    \end{cases}
    \]
    If $\abs{\mu - \rho} \leq \tfrac{\xi}{2}$,
    then $\abs{\mu-\ex{}{Z}} \leq \tau$.
\end{lem}
\begin{proof}
    Without loss of generality, we assume that $\rho \geq \mu$,
    since the argument for the other case is symmetric.
    Let $a = \rho - \xi$ and $b = \rho + \xi$.
    \begin{align}
        \abs{\mu-\ex{}{Z}} \leq \abs{\ex{}{(X-a)\ONE_{X < a}}}
            + \abs{\ex{}{(X-b)\ONE_{X > b}}}\label{eq:one-d-trunc-mean}
    \end{align}
    Now, we compute the first term on the right
    hand side. The second term would follow by an
    identical argument.
    \begin{align*}
        \abs{\ex{}{(X-a)\ONE_{X<a}}} &= \abs{\ex{}{(X-\mu-(a-\mu))\ONE_{X<a}}}\\
            &\leq \ex{}{\abs{X-\mu}\ONE_{X<a}} + (\abs{a-\mu})\ex{}{\ONE_{X<a}}\\
            &\leq \left(\ex{}{\abs{X-\mu}^k}\right)^{\frac{1}{k}}
                \left(\pr{}{X<a}\right)^{\frac{k-1}{k}}
                + \left(\abs{a-\mu}\right)\pr{}{X<a}\\
            &\leq \left(\frac{2}{C}\right)^{k-1}\tau
                + C\left(\frac{2}{C}\right)^k\tau\\
            &= 3\left(\frac{2}{C}\right)^{k-1}\tau
    \end{align*}
    In the above, the first inequality follows from
    linearity of expectations, triangle inequality,
    and Lemma~\ref{lem:jensen}. The second inequality
    follows from Lemma~\ref{lem:holder}, and the third
    inequality follows from the fact that
    $$\pr{}{X<a} \leq \pr{}{X<\mu-\tfrac{\xi}{2}}
        \leq \left(\frac{2}{C}\right)^k\tau^{\frac{k}{k-1}},$$
    which holds due to Lemma~\ref{lem:chebyshev}.
    Similarly, we can bound the second term in (\ref{eq:one-d-trunc-mean}) as follows:
    $$\abs{\ex{}{(X-b)\ONE_{X > b}}} \leq 4\left(\frac{2}{C}\right)^{k-1}\tau.$$
    Substituting these two values in Inequality~\ref{eq:one-d-trunc-mean},
    we get that
    $$\abs{\mu - \ex{}{Z}} \leq 7\left(\frac{2}{C}\right)^{k-1}\tau
        \leq \tau.$$
\end{proof}

The next lemma says that if we take a large number
of samples from any distribution over $\R$, whose
$k^{\text{th}}$ moment is bounded by $1$, then with
high probability, the empirical mean of the samples
lies close to the mean of the distribution.

\begin{lem}\label{lem:one-d-emp-mean-close}
    Let $\cD$ be a distribution over $\R$ with
    mean $\mu$ and $k^{\text{th}}$ moment bounded
    by $1$. Suppose $(X_1,\dots,X_n)$ are samples
    from $\cD$, where
    $$n \geq O\left(\frac{1}{\alpha^2}\right).$$
    Then with probability at least $0.9$,
    $$\abs{\frac{1}{n}\sum\limits_{i=1}^{n}{X_i}-\mu} \leq \alpha.$$
\end{lem}
\begin{proof}
    Using Lemma~\ref{lem:jensen}, we know that
    $$\ex{}{\abs{X-\mu}^2} \leq \ex{}{\abs{X-\mu}^k}^{\frac{2}{k}} \leq 1.$$
    Let $Z = \tfrac{1}{n}\sum\limits_{i=1}^{n}{X_i}$.
    Then we have the following.
    \begin{align*}
        \ex{}{\abs{Z-\mu}^2} &=
                \frac{1}{n^2}\ex{}{\abs{\sum\limits_{i=1}^{n}{X_i}-\mu}^2}\\
            &\leq \frac{1}{n^2}\ex{}{\sum\limits_{i=1}^{n}{\abs{X_i-\mu}^2}}\\
            &= \frac{1}{n^2}\sum\limits_{i=1}^{n}\ex{}{\abs{X_i-\mu}^2}\\
            &\leq \frac{1}{n}.
    \end{align*}
    Then using Lemma~\ref{lem:chebyshev}, we have
    $$\pr{}{\abs{Z-\mu} > \alpha} \leq \frac{1}{\sqrt{n}\alpha} \leq 0.9.$$
\end{proof}

\subsection{The Algorithm}

Here, we give an $\eps$-DP algorithm to estimate the
mean. The main algorithm consists of two parts: limiting
the data to a reasonable range so as to achieve
privacy, and to limit the amount of noise added for
it to get optimal accuracy; and mean estimation in
a differentially private way. We analyze the two
separately.

\subsubsection{Private Range Estimation}

Here, we explore the first part of the algorithm,
that is, limiting the range of the data privately.
We do that in a way similar to that of \cite{KarwaV18}.
To summarize, we use differentially private
histograms (Lemma~\ref{lem:priv-hist}) to find
the bucket with the largest number of points.
Due to certain moments of the distribution
being bounded, the points tend to concentrate
around the mean. Therefore, the above bucket
would be the one closest to the mean, and by
extending the size of the bucket a little, we
could get an interval that contains a large
number of points along with the mean. We show
that the range is large enough that the mean
of the distribution truncated to that interval
will not be too far from the original mean.

\begin{algorithm}[h!] \label{alg:pdpre}
\caption{Pure DP Range Estimator
    $\PDPRE_{\eps, \alpha, R}(X)$}
\KwIn{Samples $X_1,\dots,X_{n} \in \R$.
    Parameters $\eps, \alpha, R > 1$.}
\KwOut{$[a,b] \in \R$.}

Set parameters:
    $r \gets 10 / \alpha^{\frac{1}{k-1}}$ \\
\tcp{Estimate range}
Divide $[-R-2r,R+2r]$ into buckets: $[-R-2r,-R),\dots,[-2r,0),[0,2r),\dots,[R,R+2r]$\\
Run Pure DP Histogram for $X$ over the above buckets\\
Let $[a,b]$ be the bucket that has the maximum number of points\\
Let $I \gets [a-2r,b+2r]$ \\
\Return $I$ 
\end{algorithm}

\begin{thm}\label{thm:one-d-range-pdp}
    Let $\cD$ be a distribution over $\R$ with
    mean $\mu\in[-R,R]$ and $k^{\text{th}}$ moment bounded
    by $1$. Then for all $\eps > 0$ and $0 < \alpha < \tfrac{1}{16}$,
    there exists an $\eps$-DP algorithm that takes
    $$n \geq O\left(\frac{1}{\alpha} +
        \frac{\log(R\alpha)}{\eps}\right)$$
    samples from $\cD$, and outputs $I=[a,b] \subset \R$,
    such that with probability at least $0.9$,
    the following conditions all hold:
    \begin{enumerate}
        \item $b-a \in \Theta\left(\tfrac{1}{\alpha^{\frac{1}{k-1}}}\right)$.
        \item At most $\alpha n$ samples lie outside $I$.
        \item $\mu \in I$ and
            $b-\mu, \mu-a \geq \tfrac{10}{\alpha^{\frac{1}{k-1}}}$.
    \end{enumerate}
\end{thm}
\begin{proof}
    We separate the privacy and accuracy proofs
    for Algorithm~\ref{alg:pdpre} for clarity.

    \noindent \textbf{Privacy}:\\
    Privacy follows from Lemma~\ref{lem:priv-hist}
    and post-processing of the private output of
    private histograms (Lemma~\ref{lem:post-processing}).

    \noindent\textbf{Accuracy}:\\
    The first part follows because the intervals
    are deterministically constructed to have length
    $6r \in \Theta\left(\tfrac{1}{\alpha^{\frac{1}{k-1}}}\right)$.

    \noindent Let $(X_1,\dots,X_n)$ be independent
    samples from $\cD$. We know from Lemma~\ref{lem:chebyshev}
    that,
    $$\pr{X\sim\cD}{\abs{X-\mu} > \frac{10}{\alpha^{\frac{1}{k-1}}}}
        \leq
        \pr{X\sim\cD}{\abs{X-\mu} > \frac{10}{\alpha^{\frac{1}{k}}}}
        \leq \frac{\alpha}{10^k}.$$
    Using Lemma~\ref{lem:chernoff-mult}, we have,
    $\pr{}{\abs{\{i:X_i \not\in [\mu-r,\mu+r]\}} > \alpha n} < 0.05,$
    because $n \geq O(1/\alpha)$. Therefore, there
    has to be a bucket that contains at least
    $0.5(1-\alpha)n \geq \tfrac{n}{4}$ points from
    the dataset, which implies that the bucket
    containing the maximum number of points has
    to have at least $\tfrac{n}{4}$ points. Now,
    from Lemma~\ref{lem:priv-hist}, we know that
    the noise added to any bucket cannot exceed
    $\tfrac{n}{16}$. Therefore, the noisy value
    for the largest bucket has to be at least
    $\tfrac{3n}{16}$. Since, all these points lie
    in a single bucket, and include points that
    are not in the tail of the distribution,
    the mean lies in either the same bucket,
    or in an adjacent bucket because
    the distance from the mean is at most $r$.
    Hence, the constructed interval of length
    $6r$ contains the mean and at least
    $1-\alpha$ fraction of the points.

    \noindent From the above, since the mean
    is at most $r$ far from at least one of
    $a$ and $b$, the end points of $I$ must
    be at least $r$ far from $\mu$.
\end{proof}

The CDP equivalent of Algorithm~\ref{alg:pdpre}
(that we call, "$\CDPRE$") could be created by
using the CDP version of private histograms as
mentioned in Lemma~\ref{lem:priv-hist}. Its
approximate DP version (which we call, "\ADPRE")
could be obtained via approximate differentially
private histograms as mentioned in the same lemma.

\begin{thm}\label{thm:one-d-range-adp-cdp}
    Let $\cD$ be a distribution over $\R$ with
    mean $\mu\in[-R,R]$ and $k^{\text{th}}$ moment bounded
    by $1$. Then for all $\eps, \delta, \rho > 0$
    and $0 < \alpha < \tfrac{1}{16}$,
    there exist $(\eps,\delta)$-DP and $\rho$-zCDP
    algorithms that take
    $$n_{(\eps,\delta)} \geq O\left(\frac{1}{\alpha} +
        \frac{\log(1/\delta)}{\eps}\right)$$
    and
    $$n_{\rho} \geq O\left(\frac{1}{\alpha} +
        \sqrt{\frac{\log(R\alpha)}{\rho}}\right)$$
    samples from $\cD$ respectively, and output
    $I=[a,b] \subset \R$, such that with probability
    at least $0.9$, the following hold.
    \begin{enumerate}
        \item $b-a \in \Theta\left(\tfrac{1}{\alpha^{\frac{1}{k-1}}}\right)$.
        \item At most $\alpha n$ samples lie outside $I$.
        \item $\mu \in I$ and
            $b-\mu, \mu-a \geq \tfrac{10}{\alpha^{\frac{1}{k-1}}}$.
    \end{enumerate}
\end{thm}

\subsubsection{Private Mean Estimation}

Now, we detail the second part of the algorithm,
that is, private mean estimation. In the previous
step, we ensured that the range of the data is large
enough that the mean of the truncated distribution
would not be too far from the original mean. Here,
we show that this range is small enough, that the
noise we add to guarantee privacy is not too large.
This would imply that the whole algorithm, whilst
being differentially private, would also be accurate
without adding a large overhead in the sample complexity.

\begin{algorithm}[h!] \label{alg:pdpodme}
\caption{Pure DP $1$-Dimensional Mean Estimator
    $\PDPODME_{\eps, \alpha, R}(X)$}
\KwIn{Samples $X_1,\dots,X_{2n} \in \R$.
    Parameters $\eps, \alpha, R > 0$.}
\KwOut{$\wh{\mu} \in \R$.}

Set parameters:
$m \gets 200\log(2/\beta)$ \qquad
    $Z \gets (X_1,\dots,X_n)$ \qquad
    $W \gets (X_{n+1},\dots,X_{2n})$
\\

\tcp{Partition the dataset, and estimate mean on each subset}
\For{$i \gets 1,\dots,m$}{
    Let $Y^i \gets (W_{(i-1)\cdot \frac{n}{k}+1},\dots,W_{i\cdot \frac{n}{k}-1})$
        and $Z^i \gets (Z_{(i-1)\cdot \frac{n}{k}+1},\dots,Z_{i\cdot \frac{n}{k}-1})$\\

    \tcp{Find small interval containing the mean and large fraction of points}
    $I_i \gets \PDPRE_{\eps,\alpha,\R}(Z_i)$ \quad
        and \quad $r_i \gets \abs{I_i}$

    \tcp{Truncate to within the small interval above}
    \For{$y \in Y^i$}{
        \If{$x \not\in I_i$}{
            Set $x$ to be the nearest end-point of $I_i$
        }
    }

    \tcp{Estimate the mean}
    $\wh{\mu_i} \gets \tfrac{1}{n}\sum\limits_{y \in Y^i}{y}
        + \Lap\left(\tfrac{mr_i}{\eps n}\right)$
}

\tcp{Median of means to select a good mean with high probability}
$\wh{\mu} \gets \Median(\mu_1,\dots,\mu_m)$

\Return \wh{\mu}

\end{algorithm}

\begin{thm}\label{thm:one-d-pdp}
    Let $\cD$ be a distribution over $\R$ with
    mean $\mu\in[-R,R]$ and $k^{\text{th}}$ moment bounded
    by $1$. Then for all $\eps,\alpha,\beta > 0$,
    there exists an $\eps$-DP algorithm that takes
    $$n \geq O\left(
        \frac{\log(1/\beta)}{\alpha^2} +
        \frac{\log(1/\beta)}{\eps \alpha^{\frac{k}{k-1}}} +
        \frac{\log(R)\log(1/\beta)}{\eps}\right)$$
    samples from $\cD$, and outputs $\wh{\mu} \in \R$,
    such that with probability at least $1-\beta$,
    $$\abs{\mu-\wh{\mu}} \leq \alpha.$$
\end{thm}
\begin{proof}
    We first prove the privacy guarantee of
    Algorithm~\ref{alg:pdpodme}, then
    move on to accuracy.

    \noindent \textbf{Privacy}:\\
    The step of finding a good interval $I_i$
    is $\tfrac{\eps}{2}$-DP by Theorem~\ref{thm:one-d-range-pdp}.
    Then the step of estimating the mean by
    adding Laplace noise is $\tfrac{\eps}{2}$-DP
    by Lemma~\ref{lem:laplacedp}. Therefore,
    by Lemma~\ref{lem:composition}, each iteration
    is $\eps$-DP. Since, we use each disjoint part
    of the dataset only once in the entire loop,
    $\eps$-DP still holds. Finally, we operate
    on private outputs to find the median, therefore,
    by Lemma~\ref{lem:post-processing}, the algorithm
    is $\eps$-DP.

    \noindent \textbf{Accuracy}:\\
    We fix an iteration $i$, and discuss the
    accuracy of that step. The accuracy of the
    rest of the iterations would be guaranteed
    in the same way, since all iterations are
    independent.
    \begin{claim}
        Fix $1 \leq i \leq m$. Then in iteration
        $i$, if
        $$\abs{Y^i} \geq O\left(\frac{1}{\alpha^2} +
            \frac{1}{\eps \alpha^{\frac{k}{k-1}}} +
            \frac{\log(R)}{\eps}\right),$$
        then with probability at least $0.7$,
        $\abs{\mu_i-\mu} \leq \alpha$.
    \end{claim}
    \begin{proof}
        We know from Theorem~\ref{thm:one-d-range-pdp}
        that with probability at least $0.9$, $\mu \in I_i$,
        such that if $I_i = [a,b]$, then
        $\mu-a,b-\mu \in \Omega\left(\tfrac{1}{\alpha^{\frac{1}{k-1}}}\right)$.
        Therefore, from Lemma~\ref{lem:one-d-trunc-mean},
        the mean of the truncated distribution (let's call it
        $\mu'_i$) will be at most $\tfrac{\alpha}{3}$ from $\mu$.
        But from Lemma~\ref{lem:one-d-emp-mean-close}, we know
        that $\abs{\mu_i-\mu'_i} \leq \tfrac{\alpha}{3}$
        with probability at least $0.9$. Finally, from
        Lemma~\ref{lem:lap-conc}, with probability at least
        $0.9$, the Laplace noise added is at most $\tfrac{\alpha}{3}$
        because we have at least
        $O\left(\tfrac{1}{\eps\alpha^{\frac{k}{k-1}}}\right)$
        samples. Therefore, by triangle inequality,
        $\abs{\mu_i-\mu} \leq \alpha$, and by the union
        bound, this happens with probability at least
        $0.7$.
    \end{proof}
    Now, by the claim above, and using Lemma~\ref{lem:chernoff-add},
    more than $\tfrac{m}{2}$ iterations should yield
    $\mu_i$ that are $\alpha$ close to $\mu$,
    which happens with probability at least $1-\beta$
    (because $m \geq O(\log(1/\beta))$).
    Therefore, the median, that is, $\wh{\mu}$ is
    at most $\alpha$ far from $\mu$ with probability
    at least $1-\beta$.
\end{proof}

\subsection{Estimating with CDP and Approximate DP}

The same algorithm could be used to get CDP
guarantees by using $\CDPRE$ instead of
Algorithm~\ref{alg:pdpre}, and using the Gaussian
Mechanism (Lemma~\ref{lem:gaussiandp}) instead
of the Laplace Mechanism (Lemma~\ref{lem:laplacedp}).
We call this algorithm $\CDPODME$.
To get approximate DP guarantees, Algorithm~\ref{alg:pdpodme},
with the exception of using $\ADPRE$ instead,
could be used, and we call this modified algorithm
$\ADPODME$.

\begin{thm}\label{thm:one-d-adp-cdp}
    Let $\cD$ be a distribution over $\R$ with
    mean $\mu\in[-R,R]$ and $k^{\text{th}}$ moment bounded
    by $1$. Then for all $\eps,\delta,\rho,\alpha,\beta > 0$,
    there exist $(\eps,\delta)$-DP and $\rho$-zCDP
    algorithms that take
    $$n_{(\eps,\delta)} O\left(
        \frac{\log(1/\beta)}{\alpha^2} +
        \frac{\log(1/\beta)}{\eps \alpha^{\frac{k}{k-1}}} +
        \frac{\log(1/\delta)\log(1/\beta)}{\eps}\right)$$
    and
    $$n_{\rho} \geq O\left(
        \frac{\log(1/\beta)}{\alpha^2} +
        \frac{\log(1/\beta)}{\sqrt{\rho} \alpha^{\frac{k}{k-1}}} +
        \frac{\sqrt{\log(R)}\log(1/\beta)}{\sqrt{\rho}}\right)$$
    samples from $\cD$ respectively,
    and output $\wh{\mu} \in \R$,
    such that with probability at least $1-\beta$,
    $$\abs{\mu-\wh{\mu}} \leq \alpha.$$
\end{thm}

\subsection{Lower Bound}

Now, we prove a lower bound that matches our
upper bound, showing that our upper bounds are
optimal.

\begin{thm}\label{thm:one-d-pdp-lb}
    Let $\cD$ be a distribution with mean $\mu\in(-1,1)$
    and $k^{\text{th}}$ moment bounded by $1$.
    Then given $\eps, \delta, \alpha > 0$, any $(\eps, \delta)$-DP
    algorithm takes
    $$n \geq \Omega\left(\frac{1}{\eps\alpha^{\frac{k}{k-1}}}\right)$$
    samples to estimate $\mu$ to within $\alpha$
    absolute error with constant probability.
\end{thm}
\begin{proof}
    We construct two dstributions that are ``close'',
    and show that any $(\eps, \delta)$-DP algorithm that
    distinguishes between them requires a large
    number of samples.
    \begin{clm}
        Let $\eps,\delta,\alpha>0$, and $\cD_1$, $\cD_2$
        be two distributions on $\R$ defined as follows.
        \[\cD_1 \equiv \pr{X\sim\cD_1}{X = 0} = 1\]
        \[
            \cD_2 \equiv
            \begin{cases}
                X = 0 &\text{ with probability } 1-p\\
                X = \tau &\text{ with probability } p
            \end{cases}
        \]
        Where in the above, $\tau > 0$, $p\tau = \alpha$
        and $\alpha^{\frac{k}{k-1}} \leq p \leq \tfrac{1}{\alpha^{\frac{k}{k-1}}}$.
        Then the following holds.
        \begin{enumerate}
            \item $\ex{X\sim\cD_2}{\abs{X-p\tau}^k} \leq 1$
            \item Any $(\eps,\delta)$-DP algorithm that can distinguish
                between $\cD_1$ and $\cD_2$ with constant
                probability requires at least
                $\tfrac{1}{\eps\alpha^{\frac{k}{k-1}}}$
                samples.
        \end{enumerate}
    \end{clm}
    \begin{proof}
        For the first part, note that $\ex{X\sim\cD_2}{X} = p\tau$.
        Then we have the following.
        \begin{align*}
            \ex{X\sim\cD_2}{\abs{X-p\tau}^k} &= p\abs{\tau-p\tau}^k
                    + (1-p)\abs{p\tau}^k\\
                &= p(1-p)\tau^k\left((1-p)^{k-1} + p^k\right)\\
                &\leq p\tau^k\\
                &\leq 1. \tag{Using our restrictions on $p,\tau$}
        \end{align*}
        Now, for the second part, we know that
        $$\abs{\ex{X\sim\cD_1}{X}-\ex{X\sim\cD_2}{X}} = \alpha.$$
        Suppose we take $n$ samples each from
        $\cD_1$ and $\cD_2$. Then by Theorem~11
        of \cite{AcharyaSZ18}, we get that
        \begin{align*}
            pn &\in \Omega\left(\frac{1}{\eps}\right)\\
            \implies n &\in \Omega\left(\frac{1}{\eps\alpha^{\frac{k}{k-1}}}\right).
        \end{align*}
        We conclude by using the equivalence of pure and approximate DP for testing problems (e.g., Lemma 5 of~\cite{AcharyaSZ18}).
    \end{proof}
    Finally, since being able to learn to within
    $\alpha$ absolute error implies distinguishing
    two distributions that are at least $2\alpha$
    apart, from the above claim, the lemma holds.
\end{proof}

\section{Estimating in High Dimensions with CDP} \label{sec:high-d-cdp}

In this section, we give a computationally
efficient, $\rho$-zCDP algorithm for estimaing
the mean of a distribution with bounded $k^{\text{th}}$
moment. The analogous $(\eps,\delta)$-DP algorithm
would be the same, with the same analysis, and
we state the theorem for it at the end.

\begin{remark}
    Throughout the section, we assume that
    the dimension $d$
    is greater than some absolute constant ($32\ln(4)$).
    If it is less than that, then we can just
    use our one-dimensional estimator from
    Theorem~\ref{thm:one-d-adp-cdp} multiple
    times to individually estimate each coordinate
    with constant multiplicative overhead in
    sample complexity.
\end{remark}

\subsection{Technical Lemmata}

Similar to our one-dimensional distribution
estimator, the idea is to aggressively truncate
the distribution around a point, and compute
the noisy empirical mean. We first have to
define what truncation in high dimensions
means. Unlike the one-dimensional case, we
define truncation by throwing out points that
do not lie inside a particular ball. We would
then output the noisy empirical mean of the
remaining points. We prove its closeness to
a second definition of truncation, whose bias
is low (as we will show later).

\begin{defn}\label{def:high-d-trunc}
    Let $\rho,x \in \R^d$, and $r > 0$. Then
    we define $\trunc(\rho,r,x)$ as follows.
    \[
        \trunc(\rho,r,x,\mu) =
        \begin{cases}
            x & \text{if } \llnorm{\rho-x} \leq r\\
            \mu
                & \text{if } \llnorm{\rho-x} > r
        \end{cases}
    \]
    Similarly, for a dataset
    $S = (X_1,\dots,X_n) \in \R^{n \times d}$,
    we define $\trunc(\rho,r,S,\mu)$ as the dataset
    $S' = (X'_1,\dots,X'_n)$, where for each
    $1 \leq i \leq n$, $X'_i = \trunc(\rho,r,X_i,\mu)$.
\end{defn}

Note that the above definition would set truncated points
to the equal to the (unknown) mean $\mu$ of the distribution. 
Since $\mu$ is unknown, an algorithm could not actually
implement this truncation. Instead, we use this as a
thought-experiment and relate this truncation to the one
where points outside the ball are simply discarded. The
reason to choose this route is that the former process,
though impossible to implement, is easy to analyse,
and it is also a simple task to relate it to the latter
process.

Now, we show that the mean of the distribution,
truncated to within a large-enough ball centered
close to the mean according to Definition~\ref{def:high-d-trunc},
does not move too far from the original mean.
Our lemmas are somewhat similar to techniques in~\cite{DiakonikolasKKLMS17} (see Lemma A.18), but their results are specific to bounded second moments ($k=2$), while we focus on the case of general $k$.

\begin{lem}\label{lem:high-d-trunc-mean}
    Let $\cD$ be a distribution over $\R^d$ with
    mean $\mu$, and $k^{\text{th}}$ moment bounded
    by $1$, where $k \geq 2$. Let
    $\rho \in \R^d$, $0 < \tau < \tfrac{1}{16}$,
    and $\xi = \tfrac{C\sqrt{d}}{\tau^{\frac{1}{k-1}}}$ for a
    constant $C > 2$.
    Let $X \sim \cD$, and $Z$ be the following random
    variable.
    \[
        Z = \trunc(\rho,\xi,X,\mu)
    \]
    If $\llnorm{\mu - \rho} \leq \tfrac{\xi}{2}$,
    then $\llnorm{\mu-\ex{}{Z}} \leq \tau$.
\end{lem}
\begin{proof}
    By self-duality of the Euclidean norm, it is sufficient to prove that for each
    unit vector $v \in \R^d$,
    $\abs{\iprod{\mu-\ex{}{Z},v}} \leq \tau$.
    Let $\gamma = \ex{}{Z}$. Then we have the
    following.
    \begin{align*}
        \abs{\iprod{\mu-\gamma,v}} &=
                \abs{\iprod{\ex{}{X-\trunc(\rho,\xi,X,\mu)},v}} \tag{Linearity of expectations}\\
            &= \abs{\ex{}{\iprod{X-\trunc(\rho,\xi,X,\mu),v}}}\\
            &= \abs{\ex{}{\iprod{X-\trunc(\rho,\xi,X,\mu),v}\ONE_{X \not\in\ball{\rho}{\xi}}}}\\
            &\leq \ex{}{\abs{\iprod{X-\trunc(\rho,\xi,X,\mu),v}}\ONE_{X \not\in\ball{\rho}{\xi}}}
                \tag{Lemma~\ref{lem:jensen}}\\
            &= \ex{}{\abs{\iprod{X-\mu,v}}\ONE_{X \not\in\ball{\rho}{\xi}}}\\
            &\leq \left(\ex{}{\abs{\iprod{X-\mu,v}}^k}\right)^{\frac{1}{k}}
                \left(\ex{}{\ONE_{X \not\in\ball{\rho}{\xi}}}\right)^{\frac{k-1}{k}}
                \tag{Lemmata~\ref{lem:jensen} and \ref{lem:holder}}\\
            &\leq 1\cdot
                \left(\pr{}{\llnorm{X-\mu} > \frac{\xi}{2}}\right)^{\frac{k-1}{k}}\\
            &\leq \tau. \tag{Lemma~\ref{lem:chebyshev-high-d}}
    \end{align*}
\end{proof}

Our next lemma shows that the empirical mean
of a set of samples from a high-dimensional
distribution with bounded $k^{\text{th}}$
moment is close to the mean of the distribution.

\begin{lem}\label{lem:high-d-emp-mean-close}
    Let $\cD$ be a distribution over $\R^d$ with
    mean $\mu$ and $k^{\text{th}}$ moment bounded
    by $1$. Suppose for $\tau>0$, $(X_1,\dots,X_n)$
    are samples from $\cD$, where
    $$n \geq O\left(\frac{d}{\tau^2}\right).$$
    Then with probability at least $0.99$,
    $$\llnorm{\frac{1}{n}\sum\limits_{i=1}^{n}{X_i}-\mu} \leq \tau.$$
\end{lem}
\begin{proof}
    Suppose $\Sigma$ is the covariance matrix of
    $\cD$. We know that $\llnorm{\Sigma} = 1$.
    Let $\wb{\mu} = \sum\limits_{i=1}^{n}{X_i}$.
    Then
    $$\ex{}{\wb{\mu}} = \mu~~~\text{and}~~~
        \ex{}{(\wb{\mu}-\mu)^T(\wb{\mu}-\mu)} = \frac{1}{n}\Sigma.$$
    Using Lemma~\ref{lem:chebyshev-high-d}, we have
    the following.
    \begin{align*}
        \pr{}{\llnorm{\wb{\mu}-\mu} > \tau} &\leq
                \left(\sqrt{\frac{d}{n\tau^2}}\right)\\
            &\leq 0.99.
    \end{align*}
\end{proof}

The last lemma shows that the distance between
the empirical means of points truncated according
to the above two schemes is small.

\begin{lem}\label{lem:high-d-schemes-close}
    Let $\cD$ be a distribution over $\R^d$ with mean
    $\mu$ and $k^{\text{th}}$ moment bounded by $1$.
    Suppose for $0<\tau<\tfrac{1}{16}$, $X=(X_1,\dots,X_n)$
    are points from $\cD$, where
    $$n \geq O\left(\frac{d}{\tau^2}\right).$$
    Let $\xi=\tfrac{C\sqrt{d}}{\tau^{\frac{1}{k-1}}}$ for $C>2$,
    and $Z=\trunc(\rho,\xi,X,\mu)$. Suppose $Y$ is the
    dataset, such that, $x \in X$ lies in $Y$ if $x \in \ball{\rho}{\xi}$.
    Then with probability at least $0.95$,
    the empirical mean of $Z$ is at most $\tau$ away
    from the empirical mean of $Y$ in $\ell_2$ norm.
\end{lem}
\begin{proof}
    First, we show that $\abs{Y} = \lambda n$ with
    probability at least $0.96$, where $\lambda \geq 0.5$.
    Using Lemma~\ref{lem:chebyshev-high-d},
    we know that with probability at least $1-\tau^{\frac{k}{k-1}}$,
    a point from $\cD$ lies in $\ball{\rho}{\xi}$.
    Now, using Lemma~\ref{lem:chernoff-add}, we have
    that the number of points in $Y$ is at least
    $0.5n$ with probability $0.96$ due to our bound on $n$.

    Let $B=\ball{\rho}{\xi}$. We have the following.
    \begin{align*}
        \llnorm{\frac{1}{m}\sum\limits_{y \in Y}{y} - \frac{1}{n}\sum\limits_{z\in Z}{z}}
            &= \llnorm{\frac{1}{m}\sum\limits_{y \in B\cap X}{y} -
                \frac{1}{n}\sum\limits_{z\in Z}{z}}\\
            &= \llnorm{\frac{1}{m}\sum\limits_{y \in B\cap X}{(y-\mu)} -
                \frac{1}{n}\sum\limits_{z\in Z}{(z-\mu)}}\\
            &= \llnorm{\frac{1}{m}\sum\limits_{z \in Z}{(z-\mu)} -
                \frac{1}{n}\sum\limits_{z\in Z}{(z-\mu)}}
                \tag{$z-\mu = \vec{0}$ if $z \not\in B$}\\
            &= \frac{n-m}{mn}\cdot \sum\limits_{z \in Z}{(z-\mu)}\\
            &= \frac{1-\lambda}{\lambda}\cdot \frac{1}{n}\sum\limits_{z \in Z}{(z-\mu)}\\
            &\leq 2\tau
                \tag{Lemmata~\ref{lem:high-d-trunc-mean} and \ref{lem:high-d-emp-mean-close}}
    \end{align*}
    In the last inequality $\tfrac{1-\lambda}{\lambda}\leq 1$
    because $\lambda \geq 0.5$.
    Therefore, by rescaling by a constant factor, and
    applying the union bound, we get the required result.
\end{proof}

\subsection{The Algorithm}

We finally state the main theorem of the section
here. Algorithm~\ref{alg:cdphdme} first computes
a rough estimate of the mean using the one-dimensional
mean estimator from Theorem~\ref{thm:one-d-adp-cdp}
that lies at most $\sqrt{d}$ from $\mu$. Then it
truncates the distribution to within a small
ball around the estimate, and uses the Gaussian
Mechanism (Lemma~\ref{lem:gaussiandp}) to output
a private empirical mean.

\begin{algorithm}[h!] \label{alg:cdphdme}
\caption{CDP High-Dimensional Mean Estimator
    $\CDPHDME_{\rho, \alpha, R}(X)$}
\KwIn{Samples $X_1,\dots,X_{2n} \in \R^{d}$.
    Parameters $\rho, \alpha, R > 0$.}
\KwOut{$\wh{\mu} \in \R^d$.}

Set parameters:
$Y \gets (X_1,\dots,X_n)$ \qquad $Z \gets (X_{n+1},\dots,X_{2n})
    \qquad r \gets \tfrac{4\sqrt{d}}{\alpha^{\frac{1}{k-1}}}$

\tcp{Obtain a rough estimate of the mean via
    coordinate-wise estimation}
\For{$i \gets 1,\dots,d$}{
    $c_i \gets \CDPODME_{\tfrac{\rho}{d},1,\tfrac{0.1}{d},R}(Y^{i})$
}
Let $\vec{c} \gets (c_1,\dots,c_d)$

\tcp{Truncate to within a small ball around the mean}
Let $Z'' \gets \emptyset$\\
\For{$z \in Z$}{
    \If{$z \in \ball{\vec{c}}{r}$}{
        $Z'' \gets Z'' \cup \{z\}$
    }
}

\tcp{Recenter the data points}
Let $Z' \gets \emptyset$\\
\For{$z \in Z''$}{
    $Z' \gets Z' \cup \{z-\vec{c}\}$
}

\tcp{Estimate the mean}
Let $\ell \gets \max\{\abs{Z'},\tfrac{3n}{4}\}$\\
$\wh{\mu} \gets \tfrac{1}{\ell}\sum\limits_{z \in Z'}{z}
    + \cN\left(\vec{0},\tfrac{32r^2}{9\rho n^2}\id_{d \times d}\right)$

\Return $\wh{\mu}+\vec{c}$
\end{algorithm}

\begin{thm}\label{thm:high-d-cdp}
    Let $\cD$ be a distribution over $\R^d$ with
    mean $\mu\in\ball{\vec{0}}{R}$ and $k^{\text{th}}$ moment bounded
    by $1$. Then for all $\rho,\alpha > 0$,
    there exists a polynomial-time, $\rho$-zCDP
    algorithm that takes
    $$n \geq O\left(
        \frac{d}{\alpha^2} +
        \frac{d}{\sqrt{\rho}\alpha^{\frac{k}{k-1}}} +
        \frac{\sqrt{d\log(R)}\log(d)}{\sqrt{\rho}}
        \right)$$
    samples from $\cD$, and outputs $\wh{\mu} \in \R^d$,
    such that with probability at least $0.7$,
    $$\llnorm{\mu-\wh{\mu}} \leq \alpha.$$
\end{thm}
\begin{proof}
    The proofs of privacy and accuracy (for
    Algorithm~\ref{alg:cdphdme}) are separated
    again as follows.

    \noindent \textbf{Privacy:}\\
    Using Lemmata~\ref{thm:one-d-adp-cdp} and \ref{lem:composition},
    all calls to $\CDPODME$ are together $\rho$-zCDP.
    
    Now, we have to bound the privacy of the estimation
    step. Because we throw out points that do not lie within
    $\ball{\vec{c}}{r}$, changing one data point in $Z$
    can reduce the number of points in $Z'$ by $1$, increase
    it by $1$, or let it remain the same. The first two cases
    are symmetric, so we will only bound the sensitivity in the
    first and the third ones.

    The first case has itself has two sub-cases: $\ell > \tfrac{3n}{4}$;
    and $\ell = \tfrac{3n}{4}$. Let $W$ be the neighbouring
    dataset of $Z$, such that the corresponding $W'$ has one point
    fewer than $Z'$. Suppose $Z_\ell$ is the point in $Z'$
    that is missing from $W'$. Then we have the following
    in the first sub-case.
    \begin{align*}
        \llnorm{\frac{1}{\ell}\sum\limits_{z \in Z'}{z} -
                \frac{1}{\ell-1}\sum\limits_{w \in W'}{w}}
                &= \llnorm{\frac{(\ell-1)\sum\limits_{z \in Z'}{z} -
                \ell\sum\limits_{w \in W'}{w}}{(\ell-1)\ell}}\\
            &= \llnorm{\frac{(\ell-1)Z_\ell-\sum\limits_{w \in W'}{w}}{(\ell-1)\ell}}\\
            &= \llnorm{\frac{Z_\ell}{\ell}-\frac{1}{(\ell-1)\ell}\sum\limits_{w\in W'}{w}}\\
            &\leq \llnorm{\frac{Z_\ell}{\ell}}+
                \llnorm{\frac{1}{(\ell-1)\ell}\sum\limits_{w\in W'}{w}}\\
            &\leq \frac{r}{\ell} + \frac{(\ell-1)r}{(\ell-1)\ell}\\
            &\leq \frac{8r}{3n}
    \end{align*}
    In the second sub-case, we have the following.
    \begin{align*}
        \llnorm{\frac{4}{3n}\sum\limits_{z \in Z'}{z} -
                \frac{4}{3n}\sum\limits_{w \in W'}{w}}
                &= \frac{4}{3n}\llnorm{Z_\ell}\\
            &\leq \frac{4r}{3n}
    \end{align*}
    Therefore, by Lemma~\ref{lem:gaussiandp}, this step is
    $\rho$-zCDP. This completes the first case.

    The third case also has two sub-cases: the
    replaced point and the replacement are in $Z'$; and
    both these points are outside $Z'$. In the latter,
    the replaced point has no effect on the empirical mean,
    therefore, the sensitivity is $0$. In the former, the
    replaced point can move the empirical mean by at most
    $\tfrac{2r}{\ell} \geq \tfrac{8r}{3n}$. Therefore,
    by Lemma~\ref{lem:gaussiandp}, this step is $\rho$-zCDP.

    The final privacy guarantee follows from Lemma~\ref{lem:composition}.

    \noindent \textbf{Accuracy:}\\
    The first step finds a centre $c_i$ for each
    coordinate $i$, such that the $i^{\text{th}}$
    coordinate of the mean is at most $1$ far from
    $c_i$ in absolute distance. Therefore, $\vec{c}$
    is at most $\sqrt{d}$ away from $\mu$. By
    Theorem~\ref{thm:one-d-adp-cdp}, this happens
    with probability at least $0.9$.

    Now, by Lemma~\ref{lem:high-d-trunc-mean},
    the mean of the distribution truncated with respect to
    $\trunc$ around
    $\vec{c}$ (that we call $\mu'$) is at most
    $\alpha$ far from $\mu$ is $\ell_2$ distance.
    Therefore, by Lemmata~\ref{lem:high-d-emp-mean-close}
    and \ref{lem:high-d-schemes-close},
    the empirical mean of the truncated distribution
    with respect to the scheme used in the algorithm
    ($\wb{\mu}$) will be at most $2\alpha$ far from
    $\mu'$ with probability at least $0.9$.

    Finally, let $z=(z_1,\dots,z_d)$ be the noise
    vector added in the estimation step.
    and let
    $S_z = \sum\limits_{i \in [d]}{z_i}$. Then
    since $d \geq 32\ln(4)$, by Lemma~\ref{lem:gauss-conc},
    we have the following.
    \begin{align*}
        \pr{}{\abs{\sum\limits_{i=1}^{d}{z_i^2} -
            \frac{32dr^2}{9\rho n^2}} \geq 0.5\times\frac{32dr^2}{9\rho n^2}} \leq 0.1.
    \end{align*}
    Therefore, it is enough to have the following.
    \begin{align*}
        \frac{16dr^2}{3\rho n^2} &\leq \alpha^2\\
        \iff n &\geq \sqrt{\frac{256d^2}{3\rho\alpha^{\frac{2k}{k-1}}}}\\
            &= \frac{16d}{\sqrt{3\rho}\alpha^{\frac{k}{k-1}}}.
    \end{align*}
    This is what we required in our sample complexity.
    Hence, by the union bound and rescaling
    $\alpha$ by a constant, we get the required
    result.
\end{proof}

To get the analogous $(\eps,\delta)$-DP algorithm,
we just use $\ADPODME$ instead of $\CDPODME$ in
the first step, and keep the rest the same.

\begin{thm}\label{thm:high-d-adp}
    Let $\cD$ be a distribution over $\R^d$ with
    mean $\mu\in\ball{\vec{0}}{R}$ and $k^{\text{th}}$ moment bounded
    by $1$. Then for all $\eps,\delta,\alpha > 0$,
    there exists a polynomial-time, $(\eps,\delta)$-DP
    algorithm that takes
    $$n \geq O\left(
        \frac{d}{\alpha^2} +
        \frac{d\sqrt{\log(1/\delta)}}{\eps\alpha^{\frac{k}{k-1}}} +
        \frac{\sqrt{d\log(1/\delta)}\log(d)}{\eps}
        \right)$$
    samples from $\cD$, and outputs $\wh{\mu} \in \R^d$,
    such that with probability at least $0.7$,
    $$\llnorm{\mu-\wh{\mu}} \leq \alpha.$$
\end{thm}

\section{Estimating in High Dimensions with Pure DP} \label{sec:high-d-pdp}

We prove an upper bound for mean estimation in
case of high-dimensional distributions with bounded
$k^{\text{th}}$ moment, whilst having pure DP
guarantee for our algorithm. It involves creating
a cover over $[R,R]^d$, and using the Exponential
Mechanism (Lemma~\ref{lem:exp-mechanism}) for selecting
a point that would be a good estimate for the mean
with high probability. Note that while this algorithm
achieves sample complexity that is linear in the
dimension, it is computationally inefficient.

\subsection{Technical Lemma}

We start by stating two lemmata that would be
used in the proof of accuracy of our proposed
algorithm. Unlike Lemma~\ref{lem:one-d-trunc-mean},
the first lemma says that if we truncate a
one-dimensional distribution with bounded
$k^{\text{th}}$ moment around a point that
is far from the mean, then the mean of this
truncated distribution would be far from the
said point.

\begin{lem}\label{lem:one-d-trunc-mean2}
    Let $\cD$ be a distribution over $\R$ with
    mean $\mu$, and $k^{\text{th}}$ moment bounded
    by $1$. Let $\rho \in \R$, $0 < \tau < \tfrac{1}{16}$,
    and $\xi = \tfrac{C}{\tau^{\frac{1}{k-1}}}$ for a
    universal constant $C$.
    Let $X \sim \cD$, and $Z$ be the following random
    variable.
    \[
    Z =
    \begin{cases}
        \rho - \xi & \text{if $X < \rho - \xi$}\\
        X & \text{if $\rho - \xi \leq X \leq \rho + \xi$}\\
        \rho + \xi & \text{if $X > \rho + \xi$}
    \end{cases}
    \]
    If $\rho > \mu + \tfrac{\xi}{2}$, then
    the following holds.
    \[
        \ex{}{Z} \in
        \begin{cases}
            \left[\mu-\frac{\xi}{8},\mu+\frac{\xi}{8}\right] &
                \text{if $\frac{\xi}{2} < \abs{\rho - \mu} \leq \frac{17\xi}{16}$}\\
            \left[\rho-\xi,\rho-\frac{15\xi}{16}\right] &
                \text{if $\abs{\rho - \mu} > \frac{17\xi}{16}$}
        \end{cases}
    \]
    If $\rho < \mu - \tfrac{\xi}{2}$, then
    the following holds.
    \[
        \ex{}{Z} \in
        \begin{cases}
            \left[\mu-\frac{\xi}{8},\mu+\frac{\xi}{8}\right] &
                \text{if $\frac{\xi}{2} < \abs{\rho - \mu} \leq \frac{17\xi}{16}$}\\
            \left[\rho+\frac{15\xi}{16},\rho+\xi\right] &
                \text{if $\abs{\rho - \mu} > \frac{17\xi}{16}$}
        \end{cases}
    \]
\end{lem}
\begin{proof}
    Without loss of generality, we assume that $\rho \geq \mu$,
    since the argument for the other case is symmetric.
    Let $a = \max\left\{\mu + \tfrac{\xi}{16}, \rho - \xi\right\}$,
    $b = \rho + \xi$, and
    $q = \pr{}{\abs{X-\mu} > \abs{a-\mu}}$.
    Then the highest value $\ex{}{Z}$ can achieve
    is when $1-q$ probability mass is at $a$ and
    $q$ of the mass is at $b$. This is because the
    mass that lies beyond $a$ is $q$, and
    $b$ is the highest value that $X$ can take
    because of truncation. We get the following:
    \begin{align}
        \nonumber \ex{}{Z} &\leq (1-q)a + qb\\
            &= a + q(b-a). \label{eq:one-d-trunc-mean2}
    \end{align}
    Now, there are two cases. First, when $a = \mu + \tfrac{\xi}{16}$,
    and second, when $a = \rho - \xi$. In the first
    case, since $\mu+\tfrac{\xi}{16} \geq \rho-\xi$,
    it must be the case that $\rho-\mu \leq \tfrac{17\xi}{16}$.
    So, we have the following from (\ref{eq:one-d-trunc-mean2}).
    \begin{align*}
        \ex{}{Z} &\leq \mu + \frac{\xi}{16} +
                q\left(\rho+\xi-\mu-\frac{\xi}{16}\right)\\
            &= \mu + \frac{\xi}{16} + q(\rho-\mu) +
                \frac{15q\xi}{16}\\
            &\leq \mu + \frac{\xi}{16} + 2q\xi.
    \end{align*}
    From Lemma~\ref{lem:chebyshev}, we know that
    $$q=\pr{}{\abs{X-\mu} > \frac{\xi}{16}} \leq
        \left(\frac{16}{\xi}\right)^{k} =
        \left(\frac{16}{C}\right)^{k}\tau^{\frac{k}{k-1}}.$$
    This, along with our restrictions on $C$ and
    $\tau$, gives us,
    $$\ex{}{Z} \leq \mu + \frac{\xi}{8}.$$
    We now have to show that $\ex{}{Z} \geq \mu - \tfrac{\xi}{8}$.
    We have the following:
    \begin{align*}
        \ex{}{Z} &\geq (1-q)(\mu-\frac{\xi}{16}) + q(\rho - \xi)\\
            &= \mu - \frac{\xi}{16} + q\left(\frac{\xi}{16} - \xi + \rho - \mu\right)\\
            &= \mu - \frac{\xi}{16} - \frac{15q\xi}{16} + q(\rho - \mu)\\
            &\geq \mu - \frac{\xi}{8}.
    \end{align*}

    \noindent For the second case, we have the following
    from (\ref{eq:one-d-trunc-mean2}):
    \begin{align*}
        \ex{}{Z} &\leq \rho - \xi + q(\rho+\xi-\rho+\xi)\\
            &= \rho - \xi + 2q\xi.
    \end{align*}
    From Lemma~\ref{lem:chebyshev}, we know that
    $$q=\pr{}{\abs{X-\mu} > \abs{\rho-\xi-\mu}} \leq
        \pr{}{\abs{X-\mu} > \frac{\xi}{16}} \leq
        \left(\frac{16}{\xi}\right)^{k} =
        \left(\frac{16}{C}\right)^{k}\tau^{\frac{k}{k-1}}.$$
    This gives us,
    $$\ex{}{Z} \leq \rho - \frac{15\xi}{16}.$$
    Now it is trivial to see that $\ex{}{Z} \ge \rho - \xi$
    because the minimum value $Z$ can take is $\rho - \xi$.
\end{proof}

The guarantees of the next lemma are similar
to Lemma~\ref{lem:one-d-emp-mean-close},
except that this one promises a much higher
correctness probability. It is just the
well-known median of means estimator, which
is adapted for the case of distributions
with bounded $k^{\text{th}}$ moment.

\begin{lem}[Median of Means]\label{lem:moms}
    Let $\cD$ be a distribution over $\R$ with
    mean $\mu$ and $k^{\text{th}}$ moment bounded
    by $1$. For $0 < \beta < 1$, and $\alpha > 0$,
    suppose $(X_1,\dots,X_n)$ are independent
    samples from $\cD$, such that
    $$n \geq O\left(\frac{\log(1/\beta)}{\alpha^2}\right).$$
    Let $m \geq 200\log(2/\beta)$. For $i = 1,\dots,m$, suppose
    $Y^i = \left(X_{(i-1)m+1},\dots,X_{(i-1)m+\tfrac{n}{m}}\right)$,
    and $\mu_i$ is the empirical mean of $Y^i$.
    Define $\wb{\mu} = \Median(\mu_1,\dots,\mu_m)$.
    Then with probability at least $1-\beta$,
    $$\abs{\wb{\mu} - \mu} \leq \alpha.$$
\end{lem}
\begin{proof}
    Since $\tfrac{n}{m} \geq n_k$, by
    Lemma~\ref{lem:one-d-emp-mean-close}, with
    probability at least $0.9$, for a fixed $i$,
    $\abs{\mu_i-\mu} \leq \alpha$. Now, because
    $m \geq 200\log(2/\beta)$, using Lemma~\ref{lem:chernoff-add},
    we have that with probability at least $1-\beta$,
    at least $0.8m$ of the $\mu_i$'s are at most
    $\alpha$ far from $\mu$. Therefore, their median
    has to be at most $\alpha$ far from $\mu$.
\end{proof}

\subsection{The Algorithm}

Our algorithm for estimating the mean with pure
differential privacy is a so-called, ``cover-based
algorithm''. It creates a net of points, some of which
could be used to get a good approximation for the
mean, then privately choses one of those points
with high probability. This is done by assigning
a ``score'' to each point in the net, which depends
on the dataset, such that it ensures that the
points which are good, have significantly higher
scores than the bad points. Privacy comes in for
the part of selecting a point because the score
is directly linked with the dataset. So we use
the Exponential Mechanism (Lemma~\ref{lem:exp-mechanism})
for this purpose.

This framework is reminiscent of the classic approach of density estimation via Scheff\'e estimators (see, e.g.,~\cite{DevroyeL01}), which has recently been privatized in~\cite{BunKSW19}.
In particular, in a similar way, we choose from the cover by setting up several pairwise comparisons, and privatize using the exponential mechanism in the same way as~\cite{BunKSW19}.
This is where the similarities end: the method for performing a comparison between two elements is quite different, and the application is novel (mean estimation versus density estimation).
We adopt this style of pairwise comparisons to reduce the problem from $d$ dimensions to $2$ dimensions: indeed, certain aspects of pure differential privacy are not well understood in high-dimensional settings, so this provides a new tool to get around this roadblock.
To the best of our knowledge, we are the first to use pairwise comparisons for a statistical estimation task besides general density estimation. 
Most cover-based arguments for other tasks instead appeal to uniform convergence, which is not clear how to apply in this setting when we must preserve privacy.

The first task is to come up with a good $\Score$
function. It means that it should satisfy two
properties. First: the points $O(\alpha)$ close
to $\mu$ should have very high scores, but the
ones further than that must have very low scores.
Second: the function should have low sensitivity
so that we don't end up selecting a point with
low utility ($\Score$). We create a required
function, which is based on games or ``matches''
between pairs of points, as defined below.

\begin{defn}[Match between Two Points]\label{def:match}
    Let $X^1,\dots,X^m \in \R^{n \times d}$ be
    datasets, $X$ be their concatenation, and $p,q$
    be points in $\R^d$, and $\xi > 0$. Suppose
    $Y^1,\dots,Y^m$ are the respective datasets
    after projecting their points on to the line
    $p-q$ and truncating to within $\ball{p}{\xi}$,
    and $\mu_1,\dots,\mu_m$ are the respective
    empirical means of $Y_i$'s. Let
    $\mu' = \Median(\mu_1,\dots,\mu_m)$. Then we
    define the function $\Match_{X,\xi}:\R^d\times\R^d\rightarrow\R$
    as follows.
    \[
        \Match_{X,\xi}(p,q) =
        \begin{cases}
            \Tie & \text{if $\llnorm{p-q} \leq 20 \alpha$}\\
            \Win & \text{if $\llnorm{p-\mu'} < \llnorm{q-\mu'}$}\\
            \Lose & \text{if $\llnorm{p-\mu'} \geq \llnorm{q-\mu'}$}\\
        \end{cases}
    \]
\end{defn}
\noindent Note that the above definition is not
symmetric.

\begin{defn}[$\Score$ of a Point]
    Let $X^1,\dots,X^m \in \R^{n \times d}$ be
    datasets, $X$ be their concatenation, and $p$
    be a point in $\R^d$, and $\alpha,\xi > 0$. We define
    $\Score_{X,\xi,D,\alpha}(p)$ with respect to a domain
    $D \subset \R^d$ of a point $p$ to be the minimum
    number of points of $X$ that need to be changed
    to get a dataset $\wb{X}$
    so that there exists $q \in D$, such that
    $\Match_{\wb{X},\xi}(p,q) = \Lose$. If for
    all $q \in D \setminus \{p\}$ and all
    $Y \in \R^{nm \times d}$, $\Match_{Y,\xi}(p,q) \neq \Lose$,
    then we define $\Score_{Y,\xi,D,\alpha}(p) = n\alpha$.

    \noindent If the context is clear, we abbreviate the
    quantity to $\Score_X(p)$.
\end{defn}

\noindent By the above definition, if there
already exists a $q \in D$, such that
$\Match_{X,\xi}(p,q) = \Lose$, then $\Score_X(p) = 0$.
Let $S$ be the set as defined in Algorithm~\ref{alg:pdpodme}.
We start by showing that the points in $S$
that are close to $\mu$ have a high score
with high probability.

\begin{algorithm}[h!] \label{alg:pdphdme}
\caption{Pure DP High-Dimensional Mean Estimator
    $\PDPHDME_{\eps, \alpha, R}(X)$}
\KwIn{Samples $X_1,\dots,X_{2n} \in \R^{d}$.
    Parameters $\eps, \alpha, R > 0$.}
\KwOut{$\wh{\mu} \in \R^d$.}

Set parameters: $\xi \gets \tfrac{60}{\alpha^{\frac{1}{k-1}}}$
    \qquad $Y \gets (X_1,\dots,X_n)$
    \qquad $Z \gets (X_{n+1},\dots,X_{2n})$

\tcp{Reduce search space}
\For{$i \gets 1,\dots,d$}{
    $c_i \gets \PDPODME_{\frac{\eps}{d}.\alpha,R}(Y^i)$\\
    $I_i \gets [c_i-\alpha,c_i+\alpha]$\\
    $J_i\gets\{c_i-\alpha,c_i-\alpha+\tfrac{\alpha}{\sqrt{d}},
        \dots,c_i+\alpha-\frac{\alpha}{\sqrt{d}},c_i+\alpha\}$\\
}

\tcp{Find a good estimate}
Let $S \gets J_1 \times \dots \times J_d$\\
Compute $\Score_{X,\xi,S,\alpha}(p)$ with respect to $Z$ for every $p \in S$\\
Run Exponential Mechanism w.r.t. $\Score$
    (sensitivity $1$, privacy budget $\eps$) to output
    $\wh{\mu} \in S$ \\
\Return $\wh{\mu}$
\end{algorithm}

\begin{lem}\label{lem:high-d-pdp-case1}
    Let $m \geq 400d\log(8\sqrt{d}/\beta)$ be the same
    quantity as in Definition~\ref{def:match}, and
    let $S_\leq \subset S$, such that for all $p \in S_\leq$,
    $\llnorm{p-\mu} \leq 5 \alpha$. If
    $$n \geq O\left(\frac{m}{\alpha^2}\right),$$
    then
    $$\pr{}{\exists p \in S_\leq, \text{ st. }
        \Score_{X,\xi,S,\alpha}(p) < \tfrac{4n\alpha}{5\xi}} \leq \beta.$$
\end{lem}
\begin{proof}
    Fix a $p \in S_\leq$.
    First, note that $p$ cannot lose to any
    other point that is at most $5\alpha$ far
    from $\mu$. So, it can only lose or win
    against points that are at least $15\alpha$
    far from $\mu$.

    Now, let $q \in S$ be any point that is at
    least $20\alpha$ away from $p$, and let $\ell_q$
    be the line $p-q$. If we project $\mu$ on
    to $\ell_q$, the projected mean $\mu_0$ will be
    at most $5\alpha$ away from $p$ as projection
    cannot increase the distance between the projected
    point any of the points on the line. This
    implies that $\mu_0$ will be at least $15\alpha$
    far from $q$.
    
    From Lemma~\ref{lem:one-d-trunc-mean},
    we know that the mean of the distribution
    truncated around $p$ for a sufficiently
    large $\xi$ (which we call $\mu_p$) is at
    most $\alpha$ far from $\mu_0$. So, by
    Lemma~\ref{lem:moms}, we know that with
    probability at least $1-\tfrac{\beta}{\abs{S}^2}$,
    the median of means ($\mu'$) is at most $\alpha$
    far from $\mu_p$. This implies that $\mu'$
    is more than $14\alpha$ far from $q$, and
    at most $6\alpha$ far from $p$. Therefore,
    to lose to $q$, $\mu'$ will have to be moved
    by at least $4\alpha$ towards $q$.

    Since moving a point in $X$ can move $\mu'$
    by at most $\tfrac{2m\xi}{n}$, it means that
    we need to change at least $\tfrac{2n\alpha}{m\xi}$
    points each from at least $0.4m$ of the sub-datasets
    in $X$ to make $p$ lose to $q$ (from the proof
    of Lemma~\ref{lem:moms}), since we want to have
    the means of at least half of the sub-datasets
    to be closer to $q$.
    Taking the union bound over all pairs of
    points in $S$, we get the desired error
    probability bound. This proves the claim.
\end{proof}

Now, we prove that the points in $S$ that are
far from $\mu$ have a very low score.

\begin{lem}\label{lem:high-d-pdp-case2}
    Let $m \geq 400d\log(8\sqrt{d}/\beta)$ be as in Definition~\ref{def:match}, and
    let $S_> \subset S$, such that for all $p \in S_>$,
    $\llnorm{p-\mu} > 20\alpha$. If
    $$n \geq O\left(\frac{m}{\alpha^2}\right),$$
    then
    $$\pr{}{\exists p \in S_>, \text{ st. }
        \Score_{X,\xi,S,\alpha}(p) > 0} \leq \beta.$$
\end{lem}
\begin{proof}
    Fix a $p \in S_>$.
    We have to deal with two cases here. First,
    when $\llnorm{\mu - p} \leq \tfrac{\xi}{2}$,
    and when $\llnorm{\mu - p} > \tfrac{\xi}{2}$.

    For the first case,
    let $z$ be the point in $S$ that is nearest
    to $\mu$, and let $\ell_z$ be the line $p-z$.
    Suppose $\mu_1$ is the projection of $\mu$ on
    to $\ell_z$. Then $\mu_1$ will be at most
    $\alpha$ from $z$.
    By Lemma~\ref{lem:one-d-trunc-mean},
    the mean of the distribution truncated around
    $p$ (which we call $\mu_p$) will be at most
    $\alpha$ far from $\mu_1$. Then by Lemma~\ref{lem:moms},
    with probability at least $1-\tfrac{\beta}{\abs{S}^2}$,
    the median of means ($\mu'$) will be at most
    $\alpha$ far from $\mu_p$, hence, at most
    $3\alpha$ far from $z$. This implies that
    $\mu'$ will be at least $17\alpha$ far from
    $p$. Therefore, the score of $p$ will be $0$,
    since it has already lost to $z$.

    In the second case, let $\ell_\mu$ be the
    line $p-\mu$. Suppose $\mu_1$ is the mean
    of the distribution projected on to $\ell_\mu$
    and truncated around $p$, and let $z$ be
    the point in $S$ that is closest to $\mu_1$.
    If $\mu_1 = z$, then we're done because by
    Lemma~\ref{lem:one-d-trunc-mean2}, $\mu_1$
    is at least $\tfrac{15\xi}{16}$ far from
    $p$, and then by Lemma~\ref{lem:moms}, with
    probability at least $1-\tfrac{\beta}{\abs{S}^2}$,
    the median of means lies $\alpha$ close to $\mu_1$,
    and is closer to $\mu_1$ than it is to $p$.

    If not, then we have to do some more work.
    Now, let $\ell_z$ be the line $p-z$, let $\mu_z$
    be the projection of $\mu$ on to $\ell_z$,
    and let $\mu_2$ be the projection of $\mu_1$
    on to $\ell_z$. Using basic geometry, we
    have the following.
    \begin{align*}
        \frac{\llnorm{p-\mu}}{\llnorm{p-\mu_z}} &=
            \frac{\llnorm{p-\mu_1}}{\llnorm{p-\mu_2}}\\
        \iff \llnorm{p-\mu_z} &=
            \frac{\llnorm{p-\mu}\llnorm{p-\mu_2}}{\llnorm{p-\mu_1}}\\
        \implies \llnorm{p-\mu_z} &\geq
            \frac{\llnorm{p-\mu}\left(\llnorm{p-\mu_1}-\alpha\right)}
            {\llnorm{p-i\mu_1}}
            \tag{Triangle Inequality}\\
        &= \llnorm{p-\mu}\left(1-\frac{\alpha}{\llnorm{p-\mu_1}}\right)\\
        &\geq \llnorm{p-\mu}\left(1-\frac{16\alpha}{15\xi}\right)\\
        &= \llnorm{p-\mu}\left(1-\frac{16\alpha^{\frac{k}{k-1}}}{15C}\right)\\
        &\geq \frac{15\llnorm{p-\mu}}{16}
            \tag{Due to our restrictions on $C$ and $\alpha$}
    \end{align*}
    If $\llnorm{p-\mu_z} \leq \tfrac{\xi}{2}$,
    then by a similar argument as in the first
    case, $z$ wins against $p$ because the mean
    of the truncated distribution is close to
    $\mu_z$, and the empirical median of means
    is close to that mean with probability at
    least $1-\tfrac{\beta}{\abs{S}^2}$, hence,
    closer to $z$ than to $p$. If not, then by
    Lemma~\ref{lem:one-d-trunc-mean2},
    the mean of the distribution projected on to
    $\ell_z$, and truncated around $p$ will be
    at most $\tfrac{\xi}{16}$ far from $p-\xi$,
    so by Lemma~\ref{lem:moms}, the median of
    means ($\mu'$) will be at most $\alpha$ far
    from that mean with probability at least
    $1-\tfrac{\beta}{\abs{S}^2}$. This implies
    that it will be at most $\tfrac{\xi}{16} + 2\alpha$
    far from $z$, but will be at least $\frac{15\xi}{16}-2\alpha$
    from $p$, which means that $p$ will lose to
    $z$ by default.

    Taking the union bound over all 
    sources of error, and all pairs
    of points in $S$, we get the error probabiity
    of $4\beta$, which we can rescale to get the
    required bounds. 
\end{proof}

\noindent Finally, we prove that the $\Score$ function
has low sensitivity.

\begin{lem}\label{lem:score-sensitivity}
    The $\Score$ function satisfies the following:
    $$\Delta_{\Score,1} \leq 1.$$
\end{lem}
\begin{proof}
    Let $X$ be any dataset, and $p \in \R^d$
    be a point in the domain in question, and
    let $\Score_X(p)$ be the score of $p$. Suppose
    $X'$ is a neighbouring dataset of $X$. Then
    by changing a point $x$ in $X$ to $x'$ (to
    get $X'$), we can only change the score
    of $p$ by $1$. Let the median of means of
    projected, truncated $X$ be $\mu_1$, and
    that of $X'$ be $\mu_2$.

    Suppose $p$ was already losing to a point
    $q$, that is, its score was $0$. Then
    switching from $X$ to $X'$ can either imply
    that $\mu_2$ is further from $p$ than $\mu_1$
    was from $p$, or it could go further. In
    the first case, $p$ would still lose to $q$.
    In the second case, if $\mu_2$ is closer
    to $p$ than it is to $q$, then the score
    of $p$ would increase at most by $1$ because
    we can switch back to $X$ from $X'$ by
    switching one point; or $p$ could still
    be losing to $q$, in which case, the score
    wouldn't change at all.

    Now, suppose $p$ was winning against all
    points that are more than $30\alpha$ away
    from $p$ with respect to $X$. Let $q$ be
    a point that determined the score of $p$.
    If switching to $X'$ made $\mu_2$ closer
    to $p$ than $\mu_1$ was, then the score
    can only increase by $1$ because we can
    always switch back to $X$, and get the
    original score. If it moved $\mu_2$ closer
    to $q$ than $\mu_1$ was, then the score
    can only decrease by $1$. This is because
    $q$ determined $\Score_X(p)$ via some optimal
    strategy, and changing a point of $X$ cannot
    do better than that.
    Therefore, the sensitivity is $1$, as required.
\end{proof}

\noindent We can now move on to the main theorem
of the section.

\begin{thm}\label{thm:high-d-pdp}
    Let $\cD$ be a distribution over $\R^d$ with
    mean $\mu\in\ball{\vec{0}}{R}$ and $k^{\text{th}}$ moment
    bounded by $1$. Then for all $\eps,\alpha,\beta > 0$,
    there exists an $\eps$-DP algorithm that takes
    $$n \geq O\left(
        \frac{d\log(d/\beta)}{\alpha^2} +
        \frac{d\log(d/\beta)}{\eps\alpha^{\frac{k}{k-1}}} +
        \frac{d\log(R)\log(d/\beta)}{\eps}
        \right)$$
    samples from $\cD$,
    and outputs $\wh{\mu} \in \R^d$,
    such that with probability at least $1-\beta$,
    $$\llnorm{\mu-\wh{\mu}} \leq \alpha.$$
\end{thm}
\begin{proof}
    We again
    separate the proofs of privacy and accuracy
    of Algorithm~\ref{alg:pdphdme}.

    \noindent \textbf{Privacy:}\\
    The first step is $\eps$-DP from Lemma~\ref{thm:one-d-pdp},
    and from Lemma~\ref{lem:composition}.
    The second step is $\eps$-DP from
    Lemmata~\ref{lem:score-sensitivity} and
    \ref{lem:exp-mechanism}. Therefore, the
    algorithm is $2\eps$-DP.

    \noindent \textbf{Accuracy:}\\
    The first step is meant to reduce the size
    of the search space. From Lemma~\ref{thm:one-d-pdp},
    we have that for each $i$, the distance between
    $c_i$ and the mean along the $i^{\text{th}}$
    axis is at most $\alpha$. So, $I_i$ contains
    the mean with high probability, and is of
    length $2\alpha$ by construction.

    We know from Lemma~\ref{lem:exp-mechanism}
    that with high probability, the point returned
    by Exponential Mechanism has a high $\Score$.
    So, it will be enough to argue that with high
    probability, only the points in $S$, which are
    $O(\alpha)$ close to $\mu$, have a high quality
    score, while the rest have $\Score$ close to $0$.
    This exactly what we have from Lemmata~\ref{lem:high-d-pdp-case1}
    and \ref{lem:high-d-pdp-case2}. Let $\OPT_\Score(Z)$
    be the maximum score of any point in $S$.
    Then we know that $\OPT_\Score(Z) \geq \tfrac{4n\alpha}{5\xi}$,
    and that the points that have this score
    have to be at most $20\alpha$ far from $\mu$.
    From Lemma~\ref{lem:exp-mechanism}, we know
    that with probability at least $1-\beta$,
    \begin{align*}
        \Score(X,\wh{\mu}) &\geq \OPT_\Score(Z) -
                \frac{2\Delta_{\Score,1}}{\eps}(\log(\abs{S})+\log(1/\beta))\\
            &\geq \frac{4n\alpha}{5\xi} -
                \frac{2}{\eps}\left(d\log\left(4\sqrt{d}\right)+
                \log(1/\beta)\right)\\
            &\geq O\left(n\alpha^{\frac{k}{k-1}}\right).
                \tag{Because of our bounds on $n$ and $\xi$}
    \end{align*}
    Therefore, we get a point that is at most
    $20\alpha$ far from $\mu$. Rescaling $\alpha$
    and $\beta$ by constants, we get the required
    result.
\end{proof}

\section{Lower Bounds for Estimating High-Dimensional Distributions}\label{sec:lb}

\begin{theorem}
\label{thm:high-d-lb}
Suppose $\cA$ is an $(\eps,0)$-DP algorithm and $n \in \N$ is a number is such that, for every product distribution $P$ on $\R^d$ such that
$\ex{}{P} = \mu$ and $\sup_{v : \| v \|_2 = 1} \ex{}{ \langle v, P - \mu \rangle^2 } \leq 1$,
$$
\ex{X_1,\dots,X_n \sim P, \cA}{\| \cA(X) - \mu \|_2^2} \leq \alpha^2.
$$
Then $n = \Omega\left(\frac{d}{\alpha^2 \eps}\right)$.
\end{theorem}

The proof uses a standard \emph{packing argument}, which we encapsulate in the following lemma.
\begin{lem} \label{lem:packing-lemma}
Let $\mathcal{P} = \{P_1,P_2,\dots\}$ be a family of distributions such that, for every $P_i, P_j \in \mathcal{P}$, $\SD(P_i,P_j) \leq \tau.$  Suppose $\cA$ is an $(\eps,0)$-DP algorithm and $n \in \N$ is a number such that, for every $P_i \in \cP$,
$$
\pr{X_1,\dots,X_n \sim P_i, \cA}{\cA(X) = i} \geq 2/3,
$$
then $n = \Omega\left(\frac{\log|\cP|}{\tau \eps}\right)$.
\end{lem}

\begin{proof} We will define a packing as follows.  As a shorthand, define
    $Q_0 = 0$ and
\begin{equation*}
    Q_1 = 
\begin{cases}
0 & \textrm{w.p.~$1 - \frac{\alpha^2}{d}$} \\
\frac{\sqrt{d}}{\alpha} & \textrm{w.p.~$\frac{\alpha^2}{d}$}
\end{cases}
\end{equation*}
For $c \in \{0,1\}^d$, let
$$P_{c} = \bigotimes_{j = 1}^{d} Q_{c_{j}}$$
be the product of the distributions $Q_0$ and $Q_1$ where we choose each coordinate of the product based on the corresponding coordinate of $c$.

Note that $\SD(Q_0,Q_1) \leq \alpha^2/d$, and therefore, for every $c, c' \in \zo^d$, $\SD(P_c,P_{c'}) \leq \alpha^2$.  Let $\cC \subseteq \{0,1\}^d$ be a code of relative distance $1/4$.   That is, every distinct $c, c' \in \cC$ differ on at least $d/4$ coordinates. By standard information-theoretic arguments, there exists such a code such that $|\cC| = 2^{\Omega(d)}$.  We will define the packing to be $\cP = \{ \cP_{c} \}_{c \in \cC}$.  By Lemma~\ref{lem:packing-lemma}, if there is an $(\eps,0)$-DP algorithm $\cA$ that takes $n$ samples from an arbitrary one of the distribution $P_{c} \in \cP$ and correctly identifies $P_{c}$ with probability at least $2/3$, then $n = \Omega\left( \frac{d}{\alpha^2 \eps} \right).$

We make two more observation about the distributions in $\cP$.  First, since this is a product distribution, its 2nd moment is bounded by the maximum 2nd moment of any coordinate, so
$$
\sup_{v : \| v \|_2 = 1} \ex{}{ \langle v, P - \mu \rangle^2 } 
\leq \max\left\{ \var{}{Q_0}, \var{}{Q_1} \right\} \leq 1.
$$
Second, since any distinct $c,c'$ differ on $d/4$ coordinates, and $\ex{}{Q_1 - Q_0} = \alpha/\sqrt{d}$, we have that for every distinct $c,c'$,
$$
\| \ex{}{P_{c} - P_{c'}} \|_{2} \geq \sqrt{\frac{d}{4}} \cdot \frac{\alpha}{\sqrt{d}} = \frac{\alpha}{2}.
$$

By a standard packing argument, any $(\eps,0)$-DP algorithm that takes $n$ samples from $P_c$ for an arbitrary $c \in \cC$, and correctly identifies $c$, must satisfy $n = \Omega(\frac{d}{\alpha^2 \eps})$.  Therefore, if we can estimate the mean to within $\ell_2^2$ error $< \alpha^2/64$, we can identify $c$ uniquely.  Moreover, if $\cA$ satisfies
$$
\ex{X_1,\dots,X_n \sim P, \cA}{\| \cA(X) - \ex{}{P} \|_2^2} < \alpha^2/192
$$
for every distribution $P$ with bounded $2nd$ moment, then by Markov's inequality, we have
$$
\pr{X_1,\dots,X_n \sim P, \cA}{ \| \cA(X) - \ex{}{P} \|_2^2 < \alpha^2/64} \geq 2/3.
$$
Therefore, any $(\eps,0)$-DP algorithm $\cA$ with low expected $\ell_2^2$ error must have $n = \Omega(\frac{d}{\alpha^2 \eps})$.  The theorem now follows by a change-of-variables for $\alpha$.
\end{proof}

\section*{Acknowledgments}
We would like to thank John Duchi for bringing~\cite{BarberD14} to our attention, and Argyris Mouzakis for pointing out issues with a previous version of Lemma~\ref{lem:high-d-trunc-mean}.

\noindent GK is supported by a University of Waterloo startup grant. Part of this work was done while GK was supported as a Microsoft Research Fellow, as part of the Simons-Berkeley Research Fellowship program at the Simons Institute for the Theory of Computing, and while visiting Microsoft Research Redmond.
VS and JU are supported by NSF grants CCF-1718088, CCF-1750640, CNS-1816028, and CNS-1916020.
This work was initiated while all authors were visiting the Simons Institute for the Theory of Computing.

\addcontentsline{toc}{section}{References}
\bibliographystyle{alpha}
\bibliography{biblio}

\appendix

\section{Useful Inequalities}

The following standard concentration inequalities are used frequently in this document.

\begin{lem}[Chebyshev's Inequality]\label{lem:chebyshev}
    Let $\cD$ be a distribution over $\R$ with mean
    $\mu$, and $k^{\text{th}}$ moment bounded by $M$.
    Then the following holds for any $a > 1$.
    $$\pr{X \sim \cD}{\abs{X-\mu} > aM^{\frac{1}{k}}}
        \leq \frac{1}{a^k}$$
\end{lem}

\begin{lem}[Concentration in High Dimensions \cite{ZhuJS19}]
    \label{lem:chebyshev-high-d}
    Let $\cD$ be a distribution over $\R^d$ with mean
    $\vec{0}$, and $k^{\text{th}}$ moment bounded by $M$.
    Then the following holds for any $t > 0$.
    $$\pr{X \sim \cD}{\llnorm{X} > t}
        \leq M\left(\frac{\sqrt{d}}{t}\right)^k$$
\end{lem}

\begin{lemma}[Multiplicative Chernoff]\label{lem:chernoff-mult}
    Let $X_1,\dots,X_m$ be independent Bernoulli random variables
    taking values in $\zo$. Let $X$ denote their sum and
    let $p = \ex{}{X_i}$. Then for $m \geq \frac{12}{p}\ln(2/\beta)$,
    $$\pr{}{X \not\in \left[ \frac{mp}{2}, \frac{3mp}{2} \right]} \leq
        2e^{-mp/12} \leq \beta.$$
\end{lemma}

\begin{lemma}[Bernstein's Inequality]\label{lem:chernoff-add}
    Let $X_1,\dots,X_m$ be independent Bernoulli random variables
    taking values in $\zo$. Let $p = \ex{}{X_i}$.
    Then for $m \geq \frac{5p}{2\epsilon^2}\ln(2/\beta)$ and
    $\eps \leq p/4$,
    $$\pr{}{\abs{\frac{1}{m}\sum{X_i}-p} \geq \epsilon}
        \leq 2e^{-\epsilon^2m/2(p+\epsilon)}
        \leq \beta.$$
\end{lemma}

\begin{lem}[Laplace Concentration]\label{lem:lap-conc}
    Let $Z \sim \Lap(t)$. Then
    $\pr{}{\abs{Z} > t\cdot\ln(1/\beta)} \leq \beta$.
\end{lem}

\begin{lem}[Gaussian Empirical Variance Concentration]
    \label{lem:gauss-conc}
    Let $(X_1,\dots,X_m) \sim \cN(0,\sigma^2)$ be
    independent. If $m \geq \tfrac{8}{\tau^2}\ln(2/\beta)$,
    for $\tau\in(0,1)$, then
    $$\pr{}{\abs{\frac{1}{m}\sum\limits_{i=1}^{m}{X_i^2}-\sigma^2}
        > \tau\sigma^2} \leq \beta.$$
\end{lem}

\noindent We also mention two well-known and useful inequalities.

\begin{lem}[H\"{o}lder's Inequality]\label{lem:holder}
    Let $X,Y$ be random variables over $\R$, and
    let $k>1$. Then,
    $$\ex{}{\abs{XY}} \leq \left(\ex{}{\abs{X}^{k}}\right)^{\frac{1}{k}}
        \left(\ex{}{\abs{Y}^{\frac{k}{k-1}}}\right)^{\frac{k-1}{k}}.$$
\end{lem}

\begin{lem}[Jensen's Inequality]\label{lem:jensen}
    Let $X$ be an integrable, real-valued random
    variable, and $\psi$ be a convex function. Then
    $$\psi\left(\ex{}{X}\right) \leq \ex{}{\psi(X)}.$$
\end{lem}

\end{document}